\documentclass[11pt]{article}
\usepackage[usenames,dvipsnames]{color}
\usepackage[unicode,colorlinks=true,citecolor=black,filecolor=black,linkcolor=black,urlcolor=black, pdfpagelabels, plainpages=false]{hyperref}
\usepackage{multirow}
\usepackage
{
 amssymb, amsthm, amsmath,fullpage, nicefrac, tikz, prettyref
}
\usepackage{bbm}
\usepackage{caption}
\def\compactify{\itemsep=0pt \topsep=0pt \partopsep=0pt \parsep=0pt}
\usepackage{times}
\usepackage[margin=2.5cm]{geometry}

\newtheorem{theorem}{Theorem}
\newtheorem{lemma}[theorem]{Lemma}

\newtheorem{corollary}[theorem]{Corollary}
\newtheorem{claim}[theorem]{Claim}

\theoremstyle{definition}

\theoremstyle{remark}
\newtheorem{remark}[theorem]{Remark}

\newcommand{\I}{\mathbbm{1}}

\newcommand{\calE}{{\cal E}}
\newcommand{\calC}{{\cal C}}

\newcommand{\Exp}{\mathbb{E}}

\newcommand{\OpenFrame}{\rule{0pt}{12pt} \hrule height 0.8pt \rule{0pt}{1pt} \hrule height 0.4pt \rule{0pt}{6pt}}
\newcommand{\CloseFrame}{\vskip -8pt\rule{0pt}{1pt}\hrule height 0.4pt \rule{0pt}{1pt} \hrule height 0.8pt \rule{0pt}{12pt}}

\newcommand{\given}{\;\mid\;}

\newcommand{\fpos}{f^{+}}
\newcommand{\fneg}{f^{-}}
\newcommand{\fno}{f^{\circ}}
\newcommand{\ecost}{e.cost}
\newcommand{\elp}{e.lp}
\DeclareMathOperator{\lp}{lp}

\newcommand{\ppp}{\ensuremath{(+,+,+)}}
\newcommand{\ppm}{\ensuremath{(+,+,-)}}
\newcommand{\pmm}{\ensuremath{(+,-,-)}}
\newcommand{\mmm}{\ensuremath{(-,-,-)}}
\newcommand{\mmn}{\ensuremath{(-,-,\varnothing)}}
\newcommand{\pmn}{\ensuremath{(+,-,\varnothing)}}
\newcommand{\ppn}{\ensuremath{(+,+,\varnothing)}}
\newcommand{\fpk}{f^+_3}
\newcommand{\fmk}{f^-_3}
\newcommand{\fnk}{f^\circ_3}
\newcommand{\fp}{f^+}
\newcommand{\fm}{f^-}
\newcommand{\Cppp}{\mathcal{C}_{\ppp}}
\newcommand{\Cppm}{\mathcal{C}_{\ppm}}

\newrefformat{app}{Appendix \ref{#1}}
\newrefformat{cor}{Corollary \ref{#1}}
\newrefformat{calc}{Calculation \ref{#1}}

\begin{document}

\title{Near Optimal LP Rounding Algorithm for Correlation Clustering\\on Complete and
Complete $k$-partite Graphs}
\author{Shuchi Chawla \\ University of Wisconsin-Madison \and Konstantin Makarychev\\ Microsoft Research\and Tselil Schramm \\ UC Berkeley \and Grigory Yaroslavtsev \\ University of Pennsylvania}
\date{\textbf{\small{}}}

\maketitle

\begin{abstract}
We give new rounding schemes for the standard linear programming relaxation of
the correlation clustering problem, achieving approximation factors almost
matching the integrality gaps:
\begin{itemize}
\item For complete graphs our approximation is $2.06 - \varepsilon$, which
	almost matches the previously known integrality gap of $2$.
\item For complete $k$-partite graphs our approximation is $3$. We also show a
	matching integrality gap.
\item For complete graphs with edge weights satisfying triangle inequalities
	and probability constraints, our approximation is $1.5$, and we show an
	integrality gap of $1.2$.
\end{itemize}
Our results improve a long line of work on approximation algorithms for
correlation clustering in complete graphs, previously culminating in a ratio of
$2.5$ for the complete case by Ailon, Charikar and Newman (JACM'08).  In the
weighted complete case satisfying triangle inequalities and probability
constraints, the same authors give a $2$-approximation; for the bipartite case,
Ailon, Avigdor-Elgrabli, Liberty and van Zuylen give a $4$-approximation
(SICOMP'12).
\end{abstract}

\setcounter{page}{0}
\thispagestyle{empty} 
\pagebreak

\tableofcontents

\pagebreak

\section{Introduction}
We study the \textit{correlation clustering} problem -- given inconsistent
pairwise similarity/dissimilarity information over a set of objects, our goal
is to partition the vertices into an \textit{arbitrary} number of clusters
that match this information as closely as possible.
The task of clustering is made interesting by the fact that the similarity
information is inherently noisy. For example, we may be asked to cluster $u$
and $v$ together, and $v$ and $w$ together, but $u$ and $w$ separately. In this
case there is no clustering that matches the data exactly. The optimal
clustering is the one that differs from the given constraints at the fewest
possible number of pairs, and it may use anywhere between one and $n$ clusters.
In some contexts this problem is also known as \textit{cluster editing}: given a graph
between objects where every pair deemed similar is connected by an edge, add or
remove the fewest number of edges so as to convert the graph into a collection
of disjoint cliques.

Correlation clustering is quite different from other common clustering
objectives in that the given data is qualitative (similar versus
dissimilar pairs) rather than quantitative (e.g. objects embedded in a
metric space). As such, it applies to many different problems that
arise in machine learning, biology, data mining and other areas. From
a learning perspective, correlation clustering is essentially an
agnostic learning problem: the goal is to fit a classifier from a
certain concept class (namely all clusterings) as best as possible to
noisy examples (namely the pairwise similarity information). The
correlation clustering objective has been successfully employed for a
number of learning problems, for example: coreference
resolution~\cite{CR01, CR02, MW03},
where the goal is to determine which references in a news article
refer to the same object; cross-lingual link detection~\cite{VZ07},
where the goal is to find news articles in different languages that
report on the same event; email clustering by topic or relevance; and
image segmentation~\cite{Wirth10}. In biology, the problem of clustering gene
expression patterns can be cast into the framework of correlation
clustering~\cite{DSY99, A04}. Another application, arising in data
mining, is that of aggregating inconsistent clusterings taken from
different sources~\cite{Fil03}. In this setting, the cost of an aggregate
clustering is the sum over pairs of objects of the fraction of
clusterings that it differs from. This special case of correlation
clustering is known as the \textit{consensus clustering} problem.

In many of the above applications, we have access to a binary classifier that
takes in pairs of objects and returns a ``similar'' or ``dissimilar'' label. We
can interpret this information in the form of a complete graph with edges
labeled ``$+$'' (denoting  similarity) and ``$-$'' (denoting dissimilarity).
The correlation clustering problem can then be restated as one of producing a
clustering such that most of the ``$+$'' edges are within clusters and most
``$-$'' edges cross different clusters. In some cases, it may not be possible
to compare all pairs of objects, leaving ``missing'' edges so that the underlying
graph is not complete. However, correlation clustering on general
graphs is equivalent to the multicut problem~\cite{DEFI06}, and obtaining any
constant factor approximation is Unique-Games hard~\cite{CKKRS06}.\footnote{The
best known algorithms due to Charikar, Guruswami, Wirth~\cite{CGW05} and  
Demaine, Emanuel, Fiat, Immorlica~\cite{DEFI06}
give an $O(\log n)$ approximation.}
Still, it is possible to get approximations for some practical cases with
missing edges; for example, when the underlying graph is a complete $k$-partite graph.
Ailon, Avigdor-Elgrabli, Liberty, and van
Zuylen~\cite{AALZ12} give several applications of complete bipartite
correlation clustering.

Since its introduction a decade ago by Bansal, Blum, and Chawla~\cite{BBC04},
correlation clustering has gained a lot of prominence within the theory and
learning communities (see, e.g., the survey by Wirth~\cite{Wirth10}, and
references therein) and has become one of the textbook examples in the design
of approximation algorithms (Williamson and Shmoys~\cite{WS11}
consider the correlation clustering problem with the maximization objective).
Bansal et al. gave the first constant factor approximation algorithm for the
problem on complete graphs. The factor
has since then been improved several times, culminating in a factor of~$2.5$
for complete graphs due to Ailon, Charikar, and Newman~\cite{ACN08}, which
relies on a natural LP formulation of the problem. On the other hand, the
problem is known to be APX-hard~\cite{DEFI06}, and the best known integrality
gap of the LP is $2$~\cite{CGW05}, leaving a 20\% margin for improvement.

\subsection{Our Results}
In this paper, we nearly close the gap between
the approximation ratio and the integrality gap for complete graphs and complete
$k$-partite graphs:
For the correlation clustering problem on complete graphs, we obtain a
$(2.06 - \varepsilon)$-approximation for some fixed $\varepsilon$, nearly matching an integrality gap of $2$~\cite{CGW05}.
For the correlation clustering problem on complete $k$-partite graphs,
we obtain a $3$-approximation and exhibit an integrality gap instance with
a gap of $3$.  The previously best known algorithm for the bipartite variant
of the problem due to Ailon et al.~\cite{AALZ12} gives
a $4$-approximation.

\begin{theorem}\label{thm:main-intro}
There is a deterministic polynomial-time algorithm for the Correlation
Clustering Problem that gives a $(2.06-\varepsilon)$-approximation for complete
graphs, where $\varepsilon$ is some fixed constant smaller than $0.01$, and a
$3$-approximation for complete $k$-partite graphs.
\end{theorem}

\begin{table}[t]
	\begin{center}
		\begin{tabular}{|l || c | c| c| c|}
			\hline
			\multicolumn{5}{|c|}{\bf Approximation Algorithms for Correlation Clustering} \\
			\hline
			& \textsc{Previous Factor}
			& \textsc{Our Factor}
			& \textsc{Integrality Gap}
			& \textsc{Limitation}

			\\

			\hline
			Complete
			& $\approx 10^4$~\cite{BBC04}, 4~\cite{CGW05}, 2.5~\cite{ACN08, ZHJW07}
			& $\approx 2.06$, Thm~\ref{thm:main-intro}
			& 2~\cite{CGW05}
			& 2.025, Thm~\ref{thm:limitation-complete}

			\\
			\hline
            Triangle Inequality
			& 3~\cite{GMT07}, 2~\cite{ACN08,ZHJW07}
			& 1.5, Sec~\ref{sec:triangle-inequalities}
			& 1.2, Thm~\ref{thm:integrality-gap-triangle-inequalities}
			& 1.5, Sec~\ref{sec:triangle-inequalities}

			\\
			\hline
			Bipartite
			& 11 ~\cite{A04},4~\cite{AALZ12}
			& 3, Thm~\ref{thm:main-intro}
			& 3, Thm~\ref{thm:integrality-gap-bipartite}
			& ---
			\\
			\hline
			$K$-partite
			& ---
			& 3, Thm~\ref{thm:main-intro}
			& 3, Thm~\ref{thm:integrality-gap-bipartite}
			& ---
		\\
			\hline
		\end{tabular}
	\end{center}
	\vskip -10 pt
	\caption{Previous and our approximation factors, integrality gaps and limitations of our approach.}
	\label{table:results}
	\vskip -10 pt
\end{table}

We also study a weighted variant of the problem. In this variant, each edge has
a positive weight $\lambda^+_{uv}$ and a negative weight $\lambda^-_{uv}$, the
goal is to minimize the total weight of violated edges (for details see
Section~\ref{sec:weighted}).  We show that Weighted Correlated Clustering is
equivalent to the unweighted Correlated Clustering on complete graphs if
$\lambda^+_{uv} + \lambda^-_{uv} = 1$ for every $(u,v)\in E$.  Gionis, Mannila
and Tsaparas~\cite{GMT07} introduce a natural special case of this problem
in which the negative weights satisfy the triangle inequality
(for every $u,v,w$ it holds that $\lambda^-_{uv} \le \lambda^-_{uw} +
\lambda^-_{wv}$).  For this problem -- Weighted Correlation Clustering with
Triangle Inequalities -- we give a $1.5$-approximation algorithm, improving on
the $2$-approximation of Ailon et al.~\cite{ACN08}  (see Theorem~\ref{thm:weighted-triang-ineq} in
Section~\ref{sec:weighted}). Our proof of the last
result is computer assisted.

The main technical contribution of our paper is an approach towards obtaining
a tight rounding scheme given an LP solution. At a high level our algorithm
is similar to that of Ailon et al.~\cite{ACN08}, but we perform the actual rounding decisions
in a novel way, using carefully designed functions of the LP solution to get
rounding probabilities for edges in the graph, allowing us to obtain a
near-optimal approximation ratio.
We emphasize that although we employ a lengthy and complicated analysis
to prove that our rounding scheme achieves a $(2.06 - \varepsilon)$-approximation,
our algorithm itself is very simple and runs in time $O(n^2)$ given the LP solution.
The linear programming relaxation that we use has been studied very extensively and heuristic approaches have been developed for solving it~\cite{DSW10}.

 We demonstrate our technique for correlation clustering in complete
graphs, in complete $k$-partite graphs, and in the special case of weighted edges
satisfying triangle inequality constraints.
In each case, we obtain significant improvements over the previously best known
results, and nearly or exactly match the integrality gap of the LP.
When we are unable to match the integrality gap, we prove a lower bound on the
ratio that may be achieved by any functions within our rounding scheme.  Our
results and a comparison with the previous work are summarized in
Table~\ref{table:results}.

\subsection{Our Methodology}

As is the case for many graph partitioning problems, the correlation clustering
objective can be captured in the form of a linear program over variables that
encode lengths of edges. A long edge signifies that its endpoints should be
placed in different clusters, and a short edge signifies that its endpoints
should be in the same cluster. For consistency, edge lengths must satisfy the
triangle inequality.

A natural approach to rounding this relaxation is to interpret each edge's
length, $x_{uv}$, as the probability with which it should be cut. The challenge
is to ensure consistency. For example, consider a triangle with two positive
edges and one negative edge where the negative edge has LP length
$\tfrac{1}{2}$. If we first cut the negative edge with probability $\tfrac{1}{2}$, then
in order to return a consistent clustering we are forced to cut one of the positive edges. In
this way, an independent decision to cut or not cut one edge may force a
decision on a different edge, resulting in ``collateral damage.'' Ailon,
Charikar, and Newman~\cite{ACN08} give a simple rounding
algorithm and a charging scheme that cleanly bounds the cost of this collateral
damage. Their algorithm picks a random vertex $w$ in the graph and rounds every
edge $(w,u)$ incident on this ``pivot'' with probability equal to the length of
the edge; vertices $u$ corresponding to the edges $(w,u)$ that are not cut by
this procedure form $w$'s cluster; this cluster is then removed from the graph,
and the algorithm recurses on the remaining graph. This approach gives the best
previously known approximation ratio for the correlation clustering problem on
complete graphs, a factor of $2.5$.

Our main technical contribution is a more subtle treatment of the probability
of cutting an edge:  rather than cutting edge $(u,v)$ with probability
$x_{uv}$, we cut $(u,v)$ with probability given by some function $f(x_{uv})$
(this idea was previously used by Ailon in his algorithm for ranking
aggregation~\cite{ailon2010}).
In a departure from all of the other LP-rounding algorithms for correlation
clustering, we use different rounding functions for positive and
negative edges.
Though it may at first be surprising that the latter distinction can be
helpful, we remark that positive and negative constraints do not behave
symmetrically. For example, in a triangle with two positive edges and one
negative edge, the negative edge forces an inconsistency: any clustering of
this triangle must violate at least one constraint. On the other hand, in a
triangle with two negative edges and one positive edge, there is a valid
clustering that does not violate any constraints. Thus, we see that the
positive and negative edges behave differently, and we prove that rounding
positive and negative edges of the same LP length with different probabilities
gives a correspondingly better approximation ratio.

In some cases, this distinction leads to results that run counter to our
intuition. One might expect that negative edges should be cut with higher
probability than positive edges. However, this is not always the
case.  It turns out that it helps to cut long positive edges with probability $1$,
because we can charge them to the LP. On the other hand, it pays to be careful
about cutting long negative edges, because this might cause too much collateral
damage to other edges.

Our methodology for selecting the rounding functions $f^+$ and $f^-$ is interesting in its
own right. By regarding the cost of the algorithm on each kind of triangle as a
polynomial in $x_{uv}$, $x_{vw}$, $x_{uw}$ and in $f^{\pm}(x_{uv})$, $f^{\pm}(x_{uv})$,
$f^{\pm}(x_{uv})$,  then characterizing these multivariate polynomials, we are able to
obtain analytic upper and lower bounds on $f^+$ and $\fm$. While
this does not force our choices of $\fp$ or $\fm$, it suggests natural candidate functions
that can then be further analyzed. This same worst-case polynomial
identification and bounding approach yields lower bounds on the best possible
approximation ratio that can be achieved by a similar algorithm.

\subsection{Other Related Work}

As mentioned above, there has been a series of works giving constant
factor approximations for correlation clustering in complete
graphs~\cite{BBC04, CGW05, ACN08}.  A modified version of the Ailon et
al.~\cite{ACN08} $3$-approximation algorithm can be used as a basis for parallel
algorithms~\cite{CDK14}. Van Zuylen et al.~\cite{ZHJW07} showed that the
$2.5$-approximation of Ailon et al. can be derandomized without any loss in
approximation factor. Correlation clustering on complete bipartite graphs was
first studied by Amit~\cite{A04}, who presents an $11$-approximation. This was
subsequently improved to a $4$-approximation by Ailon et al.~\cite{AALZ12}.
For complete graphs with weights satisfying triangle inequalities a
3-approximation was obtained by Gionis et al.~\cite{GMT07} and a
2-approximation by Ailon et al.~\cite{ACN08}. These prior works are summarized
in Table~\ref{table:results}.  In general graphs, the problem can be
approximated to within a factor of $O(\log n)$, and because it is equivalent to
the multicut problem, this is suspected to be the best possible~\cite{CGW05,
DEFI06}.

Bansal et al. \cite{BBC04} also studied an alternative version of the problem
in which the objective is to approximately maximize the number of edges that
the clustering gets correct: that is, the number of ``$+$'' edges inside
clusters and ``$-$'' edges going across clusters. They noted that this version
can be trivially approximated to within a factor of $2$ in arbitrary weighted
graphs, and presented a polynomial time approximation scheme for the version on
complete graphs. Subsequently, Swamy~\cite{Swamy04} and Charikar, Guruswami and
Wirth~\cite{CGW05} developed an improved SDP-based approximations for this
``MaxAgree'' problem on arbitrary weighted graphs. MaxAgree is known to be hard
to approximate within a factor of $80/79$ for both the unweighted (complete
graph) and weighted versions~\cite{CGW05, Tan08}.

A number of other variants of the correlation clustering problem have
been studied. Giotis and Guruswami~\cite{Giotis06} and Karpinski et
al.~\cite{KS09} studied the variant where the solution is stipulated
to contain only a few (constant number of) clusters, and under that
constraint presented polynomial time approximation schemes. Mathieu,
Sankur, and Schudy~\cite{MSS10} studied the online version of the
problem and give an $O(n)$-competitive algorithm, which is also the
best possible within constant factors. Mathieu and Schudy~\cite{MS10}
introduced a semi-random model, in which every
graph is generated as follows: start with an arbitrary planted partitioning
of the graph into clusters, set labels consistent with
the planted partitioning on all edges, then independently pick every edge into a random subset
of \textit{corrupted} edges with probability $p$,
and let the adversary arbitrarily change labels on the corrupted edges.
Mathieu and Schudy~\cite{MS10} showed how to get $1+o(1)$ approximation under very mild assumptions on $p$.
Their result was extended to other semi-random models in~\cite{MMV, Yudong}.
The problem was studied in another interesting stochastic model in~\cite{AL09}.

\subsection{Organization}


In \prettyref{sec:alg}, we formally describe and generalize the algorithm and
analytical framework of Ailon et al.~\cite{ACN08} to accomodate our algorithm and analysis, and lay
the groundwork for our analysis with observations about algorithms within this
framework.
In \prettyref{sec:choosing_f}, we give the analysis that leads to our choices
of rounding functions. \prettyref{sec:pic_proofs} gives an overview of the analysis
of the triples for complete correlation clustering, as well as pictoral proofs of the main
theorem. In \prettyref{sec:gaps}, we introduce new integrality gaps for the $k$-partite
case and the weighted case with triangle inequalities; in \prettyref{sec:lbd} we
prove a lower bound of $2.025$ on the approximation ratio of any algorithm
within our framework.
\prettyref{sec:weighted} contains a summary of our results for the weighted cases of the problem.
In \prettyref{sec:derand}, we give a derandomization of our algorithm.
\prettyref{sec:kpart} and \prettyref{app:complete} contain the analytical
proofs of the $3$-approximation for $k$-partite correlation clustering and the
$(2.06-\varepsilon)$-approximation for the complete case respectively.

\section{Approximation Algorithm}\label{sec:alg}
In this section, we present an approximation algorithm for the Correlation
Clustering Problem that works both for complete graphs and complete $k$-partite
graphs. For the weighted case of the problem, we refer the reader to Section~\ref{sec:weighted}.
We denote the set of positive edges by $E^+$ and the set of negative edges by $E^-$.
The algorithm is based on the approach of Ailon et al.~\cite{ACN08}. It iteratively finds
clusters and removes them from the graph. Once all vertices are clustered, the
algorithm outputs all found clusters and terminates.

Initially, the algorithm marks all vertices as active. At step $t$, it picks a
random pivot $w$ among active vertices, and then adds each active vertex $u$ to
$S_t$ with probability $(1 - p_{uw})$ independently of other vertices. Then,
the algorithm removes the cluster $S_t$ from the graph and marks all vertices in
$S_t$ as inactive.  The probability $(1- p_{uw})$ depends on the LP solution
and the type of the pair $(u,w)$: we set $p_{uw}=\fpos(x_{uw})$, if $(u,w)$ is
a positive edge; $p_{uw} = \fneg(x_{uw})$, if $(u,w)$ is a negative edge; and
$p_{uw} = \fno(x_{uw})$, if there is no edge between $u$ and $w$. Here,
$x_{uv}$ is the LP variable corresponding to the pair $(u,v)$ (we describe the
LP in a moment) and $\fpos$, $\fneg$, and $\fno$ are special functions which we
will define later. Below we give pseudo-code for the algorithm.

\OpenFrame

\noindent \textbf{Input:} Graph $G$, $LP$ solution $\{x_{uv}\}_{u,v\in V}$.

\noindent \textbf{Output:} Partitioning of vertices into disjoint sets.

\begin{itemize}\compactify
\item Let $V_0=V$ be the set of active vertices; let $t=0$.
\item while ($V_t \neq \varnothing$)
\begin{itemize}\compactify
\item Pick a pivot $w_t\in V_t$ uniformly at random.
\item For each vertex $u\in V_t$, set $p_{uw}$:
\begin{equation}\label{eq:def-p}
p_{uw}=
\begin{cases}
\fpos (x_{uw}),&\text{if } (u,v) \text{ is a positive edge;}\\
\fneg (x_{uw}),&\text{if } (u,v) \text{ is a negative edge;}\\
\fno (x_{uw}),&\text{if there is no edge between $u$ and $v$.}
\end{cases}
\end{equation}
\item For each vertex $u\in V_t$, add $u$ to $S_t$ with probability $(1 - p_{uw})$ independently of all other vertices.
\item Remove $S_t$ from $V_t$ i.e., let $V_{t+1}=V_t\setminus S_t$. Let $t=t+1$.
\end{itemize}
\item Output $S_0,\dots, S_{T}$, where $T$ is the index of the last iteration of the algorithm.
\end{itemize}
\CloseFrame

This algorithm is probabilistic. We show how to derandomize it in \prettyref{sec:derand}.
The main new ingredient of our algorithm is a procedure for picking the
probabilities $p_{uw}$.  To compute these probabilities, we use the standard LP
relaxation introduced by~\cite{CGW05}. We first formulate an integer program for
Correlation Clustering. For every pair of vertices $u$ and $v$, we have a
variable $x_{uv}\in\{0,1\}$ that is equal to the distance between $u$ and $v$
in the ``multicut metric'': $x_{uv}=0$ if $u$ and $v$ are in the same cluster;
and $x_{uv}=1$ if $u$ and $v$ are in different clusters. Variables $x_{uv}$
satisfy the triangle inequality constraints (\ref{lp:triang}). They are also
symmetric, i.e. $x_{uv}=x_{vu}$. Instead of writing the constraint
$x_{uv}=x_{vu}$, we have only one variable for each edge $(u,v)$.
We refer to this variable as $x_{uv}$ or $x_{vu}$. The IP objective
is to minimize the number of violated constraints. We write it as follows:
$\min \sum_{(u,v)\in E^+}  x_{uv} +  \sum_{(u,v)\in E^-} (1 - x_{uv})$.
Note that a term $x_{uv}$ in the first sum equals 1 if and only if the
corresponding positive edge $(u,v)$ is cut; and a term $(1-x_{uv})$ in the
second sum equals 1 if and only if the corresponding negative edge $(u,v)$ is
contracted. Thus, this integer program is exactly equivalent to the Correlation
Clustering Problem.
We relax the integrality constraints $x_{uv}\in \{0, 1\}$ and obtain the
following LP.
\begin{align}
\min \sum_{(u,v)\in E^+} &x_{uv} + \sum_{(u,v)\in E^-}(1-x_{uv})\label{lp:obj}\\
x_{uv}+x_{vw}&\geq x_{uw}&\text{for all } u,v,w\in V\label{lp:triang}\\
x_{uu}&=0&\text{for all } u\in V\\
x_{uv}&\in [0,1] &\text{for all } u,v\in V
\end{align}

The approximation ratio of the algorithm depends on the set of functions
$\{\fpos, \fneg, \fno\}$ we use for rounding. We explain how we pick
these functions in \prettyref{sec:choosing_f}. In
Section~\ref{sec:kpart} and Appendix~\ref{app:complete}, we analyze the specific
rounding functions that give improved approximation guarantees for complete
graphs and complete $k$-partite graphs. We prove the following theorems.

\begin{theorem}[See \prettyref{app:complete}]\label{thm:complete}
For complete graphs, the approximation algorithm with rounding functions
\[
f^+(x) = \begin{cases} 0, &\text{if } x < a \\
\Bigl(\frac{x-a}{b-a}\Bigr)^2,&\text{if } x\in [a,b]\\
1 &\text{if } x \ge b\end{cases},
\qquad
f^-(x) = x,
\]
gives a $(2.06 - \varepsilon)$-approximation for $a=0.19$ and
$b=0.5095$, and a constant $\varepsilon$ with $0<\varepsilon<0.01$.
\end{theorem}

\begin{theorem}[See \prettyref{sec:kpart}]\label{thm:kpart}
For complete $k$-partite graphs, the approximation algorithm with rounding functions
\[
f^+_3(x) = \begin{cases} 0, &\text{if } x < \frac{1}{3} \\
 1 & \text{if }x \ge \frac{1}{3} \end{cases},\qquad
f^-_3(x) = x,\qquad
\fnk(x) = \begin{cases} \frac{3}{2}x &\text{if } x \le \frac{2}{3} \\
1 &\text{if } x > \frac{2}{3} \end{cases},
\]
gives a 3-approximation.
\end{theorem}

Note that the integrality gap for complete graphs is $2$~\cite{CGW05}, and the
integrality gap for complete $k$-partite graphs is $3$ (see
Section~\ref{sec:gaps}).  So the algorithm for complete $k$-partite graphs
optimally rounds the LP; the algorithm for complete graphs nearly optimally
rounds the LP.

\subsection{Analysis}
In this section, we prove a general statement -- Lemma~\ref{lem:alpha} -- that
asserts that the approximation ratio of the algorithm is at most $\alpha$ if a
certain condition (depending on $\alpha$) holds for every triple of vertices
$u,v,w\in V$.  We shall assume that $f(0) = 0$ for each $f\in \{\fpos, \fneg,
\fno\}$, and, particularly, that the algorithm always puts the pivot $w_t$ in
the set $S_t$.

Consider step $t$ of the algorithm. At this step, the algorithm finds and
removes set $S_t$ from the graph. Observe, that if $u\in S_t$ or $v\in S_t$,
then the constraint $(u,v)$ is either violated or satisfied right after step
$t$. Specifically, if $(u,v)$ is a positive edge, then the constraint $( u,v)$
is violated if exactly one of the vertices -- $u$ or $v$ -- is in $S_t$. If
$(u,v)$ is a negative constraint, then $(u,v)$ is violated if both $u$ and $v$
are in $S_t$. Denote the number of violated constraints at step $t$ by $ALG_t$.
Then,
$$ALG_t = \sum_{\substack{(u,v)\in E^+\\u,v\in V_t}}
(\I(u\in S_t;v\notin S_t) +\I(u\notin S_t;v\in S_t)) +
\sum_{\substack{(u,v)\in E^-\\u,v\in V_t}}\I(u\in S_t;v\in S_t).
$$
Here $\I(\calE)$ denotes the indicator function of the event $\calE$.  We want
to charge the cost of constraints violated at step $t$ to the LP cost of edges
removed at step $t$. The LP cost of edges removed at this step equals
$$LP_t = \sum_{\substack{(u,v)\in E^+\\u,v\in V_t}} \I (u\in S_t\text{ or }v\in S_t) x_{uv} +
\sum_{\substack{(u,v)\in E^-\\u,v\in V_t}} \I (u\in S_t\text{ or }v\in S_t) (1 - x_{uv}).$$
Note, that $\sum_{t=0}^T LP_t = LP$, since every edge is removed from the graph
exactly once (compare the expression above with the objective
function~(\ref{lp:obj})).  If we show that $\Exp[ALG_t] \leq \alpha \Exp
[LP_t]$ for all $t$, then we will immediately get an upper bound on the
expected total cost of the clustering:
$$\Exp [ALG] = \Exp \Big[\sum_{t=0}^{T}  ALG_t\Big]\leq
\alpha \Exp \Big[\sum_{t=0}^{T} LP_t\Big] = \alpha LP.$$
Here, we use that $X_s = \sum_{t=0}^{s} \big(\alpha LP_t - ALG_t\big)$ is a
submartingale (i.e., $\Exp [X_{s+1}\given X_s]\geq X_s$) , and $T$ is a stopping time.

Let $\ecost_w(u,v)$ be the conditional probability of violating the constraint
$(u,v)$ given that the pivot $w_t$ is $w$ (assuming $u,v,w\in V_t$, and $u\neq
v$). We say that $\ecost_w(u,v)$ is the expected cost of the constraint $(u,v)$
given pivot $w$. Similarly, let $\elp_w(u,v)$ be the conditional probability of
removing the edge $(u,v)$ given the pivot is $w$ multiplied by the LP cost of
the edge $(u,v)$. (The LP cost equals $x_{uv}$ for positive edges; and
$(1-x_{uv})$ for negative edges.) That is,
\begin{align}
\ecost_w(u,v) &=
\begin{cases}
p_{uw}(1-p_{vw}) + (1-p_{uw})p_{vw},&\text{if }(u,v)\in E^+;\\
(1-p_{uw})(1-p_{vw}),&\text{if }(u,v)\in E^-;\\
0,&\text{if }(u,v)\notin E;
\end{cases}
\label{eq:ecost}\\
\elp_w(u,v) &=
\begin{cases}
(1-p_{uw}p_{vw})x_{uv},&\text{if }(u,v)\in E^+;\\
(1-p_{uw}p_{vw})(1-x_{uv}),&\text{if }(u,v)\in E^-;\\
0,&\text{if }(u,v)\notin E.
\end{cases}
\label{eq:elp}
\end{align}
The expressions above do not depend on the set of active vertices $V_t$.
The cut probabilities $p_{uw}$ and $p_{vw}$ are defined by the algorithm
(see equation~(\ref{eq:def-p})). Note
that $\ecost_u(u,v)$, $\ecost_v(u,v)$, $\elp_u(u,v)$, and $\elp_v(u,v)$ are
well defined. We also formally define $\ecost_w(u,u)$ and $\elp_w(u,u)$ using
formulas~(\ref{eq:ecost}) and~(\ref{eq:elp}). In the analysis of the algorithm
for complete graphs, we assume that each vertex $u$ has a positive self-loop,
thus $\ecost_w(u,u) = 2(1-p_{uw})p_{uw}$, $\elp_w(u,u) = 0$.

We now write $\Exp[ALG_t]$ and $\Exp[LP_t]$ in terms of $\ecost$ and $\elp$:
\begin{align*}
\Exp[ALG_t\given V_t] &= \sum_{\substack{(u,v)\in E\\u,v\in V_t}}
\Big(\frac{1}{|V_t|}\sum_{w\in V_t} \ecost_w(u,v)\Big)
= \frac{1}{2|V_t|}\sum_{u,v,w \in V_t:u\neq v}\ecost_w(u,v);\\
\Exp[LP_t\given V_t] &= \sum_{\substack{(u,v)\in E\\u,v\in V_t\\}}
\Big(\frac{1}{|V_t|}\sum_{w\in V_t} \elp_w(u,v)\Big)
= \frac{1}{2|V_t|}\sum_{u,v,w \in V_t:u\neq v}\elp_w(u,v).
\end{align*}
We divided the expressions on the right hand side by 2, because, in the sum, we
count every $\ecost_w(u,v)$ and $\elp_w(u,v)$ twice (e.g., the first sum
contains the terms $\ecost_w(u,v)$ and $\ecost_w(v,u)$).  We now add terms
$\ecost_w(u,u)$ to the first sum and terms $\elp_w(u,u)$ to the second sum.
Then, we group all terms containing $u$, $v$, and $w$ together. Note that
$\ecost_w(u,u)\geq 0$ and $\elp_w(u,u)=0$. We get
\begin{align}
\Exp[ALG_t\given V_t] &\leq \frac{1}{6|V_t|}\sum_{u,v,w \in V_t}
\underbrace{\ecost_w(u,v) + \ecost_v(w,u) + \ecost_u(v,w)}_{ALG(uvw)};\label{def:ALGuvw}\\
\Exp[LP_t\given V_t] &= \frac{1}{6|V_t|}\sum_{u,v,w \in V_t}
\underbrace{\elp_w(u,v) + \elp_v(w,u) + \elp_u(v,w)}_{LP(uvw)}\label{def:LPuvw}.
\end{align}
We denote each term in the first sum by $ALG(uvw)$ and each term in the second
sum by $LP(uvw)$. Observe, that if $ALG(uvw)\leq \alpha LP(uvw)$
for all $u,v,w\in V$, then
$\Exp[ALG_t] \leq \alpha \Exp[LP_t]$, and, hence,
$\Exp[ALG] \leq \alpha LP$. We thus obtain the following lemma.

\begin{lemma}\label{lem:alpha}
Fix a set of functions $\{\fpos, \fneg, \fno\}$ with $\fpos(0) = \fneg(0) =
\fno(0)=0$.  If $ALG(uvw)\leq \alpha LP(uvw)$ for every $u,v,w\in V$ (see
(\ref{eq:ecost}), (\ref{eq:elp}), (\ref{def:ALGuvw}), and (\ref{def:LPuvw}) for
definitions), then the expected number of violated constraints at step $t$ is
bounded by $\alpha$ times the expected LP volume removed at step $t$:
$$\Exp[ALG_t] \leq \alpha \Exp[LP_t].$$
Consequently, the expected cost of the clustering returned by the algorithm is
upper bounded by $\alpha LP$.
\end{lemma}

\subsection{Triple-Based Analysis}

\label{sec:performing_analysis}
To finish analysis we need to show that $ALG(uvw)\leq \alpha LP(uvw)$ for every
triple of vertices $u,v,w\in V$. We separately analyze our choice of functions
$f$ for complete graphs and complete $k$-partite graphs in
Section~\ref{app:complete} and in Section~\ref{sec:kpart} respectively.  We
show that functions from Theorem~\ref{thm:complete} satisfy the conditions of
Lemma~\ref{lem:alpha} with $\alpha = 2.06 $ for complete graphs; and functions
from Theorem~\ref{thm:kpart} satisfy the conditions of Lemma~\ref{lem:alpha}
with $\alpha = 3$ for complete $k$-partite graphs. To show that $ALG(uvw) \leq
\alpha LP(uvw)$ for every triangle $uvw$, we consider all triangles $uvw$
with LP values satisfying triangle inequalities with all possible types of
edges: positive, negative, and ``missing'' or ``neutral'' edges. For brevity,
we refer to triangles as $(s_{uv},s_{vw},s_{uw})$ where each $s$ is one of the
symbols ``$+$'', ``$-$'' or ``$\varnothing$''.  For example, a
$(+,-,\varnothing)$-triangle is a triangle having two edges: a positive edge
and a negative edge; the third edge is missing.

The analysis of the functions $f$ requires considering many cases. We
show that $ALG(uvw)\leq \alpha LP(uvw)$ for all
$(s_1,s_2,s_3)$-triangles with edge lengths $(x,y,z)$ satisfying
triangle inequalities that may possibly appear in each type of graph
(we use that there are no neutral edges in a complete graph, and that
no triangle in a complete $k$-partite graph contains exactly two
neutral edges).  In other words, we fix types of edges for every
triangle and then consider a function of edge lengths
$x_{uv}$,$x_{vw}$, $x_{uw}$, and cut probabilities $p_{uv}$, $p_{vw}$,
$p_{uw}$ formally defined using algebraic expressions
(\ref{eq:ecost}), (\ref{eq:elp}), (\ref{def:ALGuvw}), and
(\ref{def:LPuvw}): $$\calC(x_{uv}, x_{vw}, x_{uw}, p_{uv},
p_{vw}, p_{uw}) = \alpha LP(uvw) - ALG(uvw).$$ Note, that this
function is defined even for those triangles that are not present in
our graph. We show that
$$\calC(x_{uv}, x_{vw}, x_{uw}, f^{uv}(x_{uv}), f^{vw}(x_{vw}), f^{uw}(x_{uw}))
\geq 0,$$
for all $x_{uv}, x_{vw}, x_{uw}\in [0,1]$
satisfying the triangle
inequality. Here, $f^{uv}$, $f^{vw}$ and $f^{uw}$ are rounding functions for
the edges $(u,v)$, $(v,w)$ and $(u,w)$ respectively. We first prove that for
many rounding functions $f$, and, particularly, for rounding functions we use
in this paper, it is sufficient to verify the inequality $\calC\geq 0$ only for
those $x_{uv}$, $x_{vw}$, and $x_{uw}$ for which the triangle inequality is
tight, and a few corner cases.

\begin{lemma}\label{lem:tight-triangle}
Suppose that $\fpos$ is a monotonically non-decreasing piecewise convex
function; $\fneg$ is a monotonically non-decreasing piecewise concave function;
and $\fno$ is a monotonically non-decreasing function. Let $A^+$ be the set of
endpoints of convex pieces of $\fpos$, and $A^-$ be the set of endpoints of the
concave pieces of $f^-$. (Note that $\{0,1\}\subset A^+$ and $\{0,1\}\subset
A^-$.) If the conditions of Lemma~\ref{lem:alpha} are violated for some
$(s_{uv}, s_{vw}, s_{uw})$-triangle with edge lengths $(x^*_{uv}, x^*_{vw},
x^*_{uw})$ i.e., $\alpha LP (uvw)- ALG(uvw) < 0$, then there exists possibly
another triangle $(x_{uv}, x_{vw}, x_{uw})$ (satisfying triangle inequalities)
for which $\alpha LP (uvw)- ALG(uvw) < 0$ such that either
\begin{enumerate}
\item the triangle inequality is tight for $(x_{uv}, x_{vw}, x_{uw})$; or
\item the lengths of all positive edges of the triangle belong to $A^+$;
the lengths of all negative edges of the triangle belong to $A^-$; the lengths
of all neutral edges belong to $\{0,1\}$.
\end{enumerate}
\end{lemma}
\begin{proof}
We show that the function $\calC(x_{uv}, x_{vw}, x_{uw}, f^{uv}(x_{uv}),
f^{vw}(x_{vw}), f^{uw}(x_{uw}))$ has the global minimum over the region
satisfying triangle inequalities at a point $(x_{uv}, x_{vw}, x_{uw})$
satisfying (1) or (2). By slightly perturbing functions $f$, we may assume that
these functions are strictly increasing.%
\footnote{Formally, we consider a sequence of strictly monotone functions
uniformly converging to our rounding functions.}
Suppose that $\calC$ has a minimum at point $x$. Consider one of the edges,
say, $x_{uv}$ whose length does not lie in the corresponding set $A^+$, $A^-$
or $\{0,1\}$. Let $p_{uv} = f^{uv}(x_{uv})$; and let
$g^{uv}(p_{uv})=(f^{uv})^{-1}(p_{uv})$. Note that the function $\calC$ is a
multilinear polynomial of $x$'s and $p$'s; in particular, it is linear in
$p_{uv}$ and $x_{uv}$. Moreover, it does not have monomials containing both
$x_{uv}$ and $p_{uv}$. This follows from the formal definition of $\calC$. Let
us fix all variables except for $x_{uv}$ and $p_{uv}$. We now express $x_{uv}$
as a function of $p_{uv}$: $x_{uv}= g^{uv} (p_{uv})$. The
function $g^{uv}$ is locally concave if $(u,v)$ is a positive edge; and it is
locally convex if $(u,v)$ is a negative edge. Here we use that $x_{uv}$ is not
in $A^+$ or $A^-$.  Observe that the only term containing $x_{uv}$ in $\calC$
comes from the expression for $LP(uvw)$. Thus, the coefficient of $x_{uv}$ is
positive if $(u,v)$ is a positive edge; and it is negative if $(u,v)$ is a
negative edge. The function $\calC$ does not depend on $x_{uv}$ if $(u,v)$ is a
neutral edge (of course, $\calC$ may depend on $p_{uv}$). So the function
$\calC(g^{uv}(p_{uv}), x_{vw}, x_{uw}, p_{uw}, f^{vw}(x_{vw}), f^{uw}(x_{uw}))$
is a concave function of $p_{uw}$ (when $x_{vw}$ and $x_{uw}$ are fixed).
Therefore, if we slightly decrease or increase $p_{uw}$ the value of $\calC$
will decrease. But we assumed that $\calC$ has the global minimum at $x$.
Hence, $x$ lies on the boundary of the region constrained by the triangle
inequality.  This concludes the proof.
\end{proof}

We give the analysis for the complete case in \prettyref{app:complete} and the
considerably simpler analysis for the $k$-partite case in
\prettyref{sec:kpart}.

\section{Choosing the Rounding Functions}\label{sec:choosing_f}
We now discuss how we came up with the rounding functions $f$ for complete graphs.
We need to find functions $f^+$ and $f^-$ that satisfy the conditions
of Lemma~\ref{lem:alpha}, i.e., such that for all triangles $uvw$,
$\alpha LP(uvw) - ALG(uvw)\geq 0$. (For brevity, we shall drop the argument $uvw$ of
the $ALG$ and $LP$ functions from now on.) Lemma~\ref{lem:tight-triangle} suggests that
we only need to care about triangles for which the triangle inequality is
tight. We focus on some special families of such triangles and obtain lower and
upper bounds of $f^+$ and $f^-$. Then, we find functions satisfying these constraints.
Later, we formally prove that the functions we found indeed give $\alpha = 2.06$
approximation. We sketch the proof in Section~\ref{sec:pic_proofs} and give a detailed
proof in \prettyref{app:complete}.

We first consider a $\pmm$-triangle with edge lengths $(0,x,x)$.
\begin{lemma}
For a $\pmm$-triangle with edge lengths $(0,x,x)$,
\[
	\frac{ALG}{LP} = \frac{1-\fm(x)^2}{1-x}.
\]
\end{lemma}

\begin{proof}
We simply calculate $ALG$ and $LP$ using \prettyref{eq:ecost} and
\prettyref{eq:elp}.
\end{proof}

\begin{corollary}\label{cor:fmbound}
Any function $\fm$ that achieves an $\alpha$-approximation on all \pmm-triangles satisfies
	\[
		\fm(x) \ge \sqrt{ 1- \alpha(1-x)}
	\]
for all $x\in[0,1]$.
\end{corollary}

\begin{claim}
	The function $\fm(x) = x$ does not violate the conditions of \prettyref{cor:fmbound} for $\alpha = 2.06$.
\end{claim}

Fixing $\alpha = 2.06$, the function we select is restricted to the blue region below.
\begin{center}
	\includegraphics[width=0.4\textwidth,natwidth=239,natheight=166]{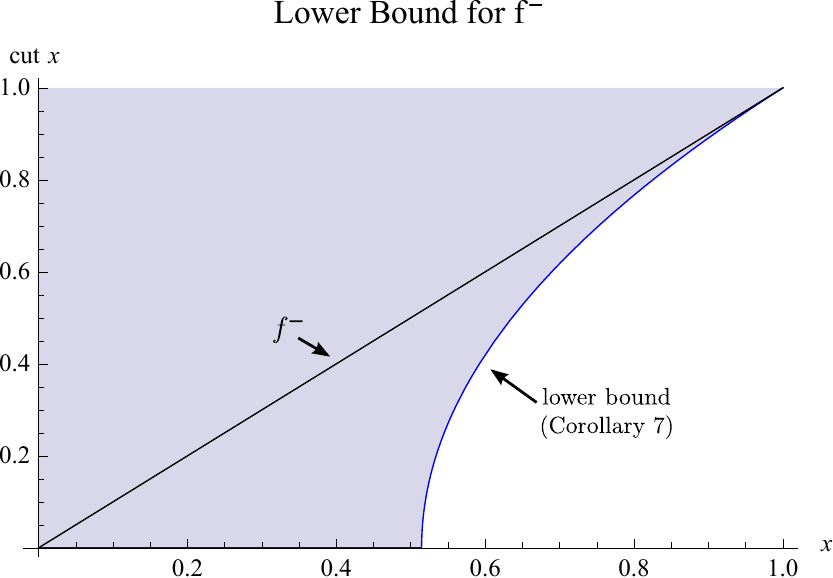}
\end{center}

Thus we take $\fm(x) = x$, as this choice is an easy candidate for the analysis.
Now, we bound $\fp$ using the tight case for the linear rounding, a
\ppm-triangle with edge lengths $(x,x,2x)$.

\begin{lemma}
	\label{lem:fp_lower}
	Any function $\fp$ that achieves an $\alpha$-approximation
        ratio on all \ppm-triangles has
\[
	\fp(x) \ge \frac{8 x-4 \alpha x^2-\sqrt{\left(4 \alpha x^2 - 8x \right)^2-4 (1-\alpha+4 x) (1+\alpha-2 \alpha x)}}{2 (1+\alpha-2 \alpha x)},
\]
for $x \in [0,1/2]$, if $\fm(x) = x$.
\end{lemma}

\begin{proof}
	Again, we calculate $\alpha LP - ALG$ using \prettyref{eq:ecost} and
	\prettyref{eq:elp}. This yields a quadratic function in $\fp(x)$:
	\begin{align*}
		\alpha LP - ALG &=-1+\alpha -4 x-4 x (-2+\alpha x) \fp(x) -(1+\alpha-2 \alpha x) \fp(x)^2.
	\end{align*}
	Solving for $\alpha LP - ALG \ge 0$ in terms of $\fp(x)$, we get our result.
\end{proof}

This lower bound on $\fp$ is necessary for approximating \ppm-triangles well,
but choosing a large $\fp$ has consequences for the approximation ratio of
$\ppp$-triangles. We use a \ppp-triangle with edge lengths $(x,x,0)$ to obtain
an upper bound on $\fp$.

\begin{lemma}
	\label{lem:fp_upper}
Any function $\fp$ that achieves an $\alpha$-approximation
ratio on all \ppp-triangles satisfies for all $x\in [0,1]$,
\[
	\fp(x) \le 1 - \sqrt{1 - \alpha x}.
\]
\end{lemma}

\begin{proof}
	We compute $\alpha LP - ALG$ using \prettyref{eq:ecost} and
	\prettyref{eq:elp} as before. This yields a different quadratic function in $\fp(x)$:
	$$\alpha LP - ALG = 2(\alpha x - 2\fp(x) + \fp(x)^2).$$
	Solving for $\alpha LP - ALG \ge 0$ in terms of $\fp(x)$, we get our result.
\end{proof}

The bounds from \prettyref{lem:fp_upper} and \prettyref{lem:fp_lower} give a
restricted region in which $\fp(x)$ may be for $x \in [0,\tfrac{1}{2}]$ and $x
\in [0,\tfrac{1}{\alpha}]$ to get an $\alpha$-approximation. We chose $\fp$
so that it would violate neither constraint, and also be easy to analyze.

\begin{claim}
Functions $f^+$ and $f^-$ from \prettyref{thm:complete} do not violate the conditions in \prettyref{lem:fp_upper} or \prettyref{lem:fp_lower} for an $\alpha = 2.06$ approximation when $a = 0.19, b = 0.5095$.
\end{claim}
\begin{center}
	\includegraphics[width=0.4\textwidth,natwidth=239,natheight=168]{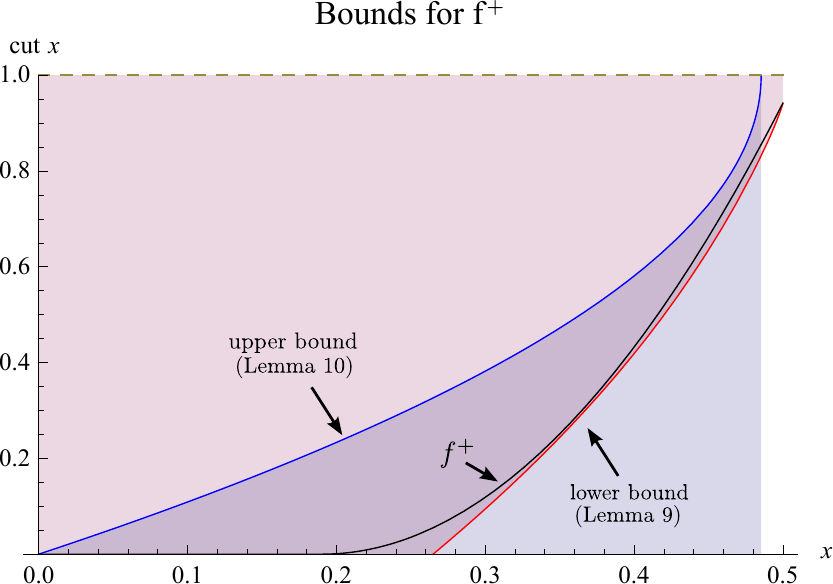}
\end{center}

The parameters
$a,b$ were chosen computationally within these analytic bounds so as to
minimize $\alpha$.

\section{Pictorial Proofs} \label{sec:pic_proofs}
Here we give a \emph{pictorial proof} of our main result, \prettyref{thm:complete}. This
proof serves as an illustration for an analytical proof we present in
\prettyref{app:complete}.

To prove \prettyref{thm:complete}, we use the framework presented in \prettyref{sec:performing_analysis}, bounding the approximation ratio of each triangle for every set of LP weights permitted by the constraints. \prettyref{lem:tight-triangle} allows us to consider only triangles for which
the triangle inequality is tight, that is, triangles of the form $(x,y,x+y)$.
\prettyref{fig:ppp_ratio}, \prettyref{fig:ppm_ratio}, \prettyref{fig:pmm_ratio}, and \prettyref{fig:mmm_ratio} are plots of the polynomials $2.06 \cdot LP(uvw) - ALG(uvw)$ when the triangle inequality is tight; the fact that each of these polynomials is positive in the range of possible LP weights proves \prettyref{thm:complete}.

In \prettyref{app:complete}, we provide an analytical proof of
\prettyref{thm:complete}. For each case, we show
that this difference polynomial is positive for all possible LP weights. In the first two cases,
we take partial derivatives to find the worst triangle lengths in terms of a single
variable, then bound the roots of the polynomials; in the latter two cases we
are able to provide a factorization for the polynomial that is positive
term-by-term. For the complete argument, see \prettyref{app:complete}.

\begin{figure}[h!]
	\begin{center}
	\includegraphics[width=0.4\textwidth,natwidth=288,natheight=240]{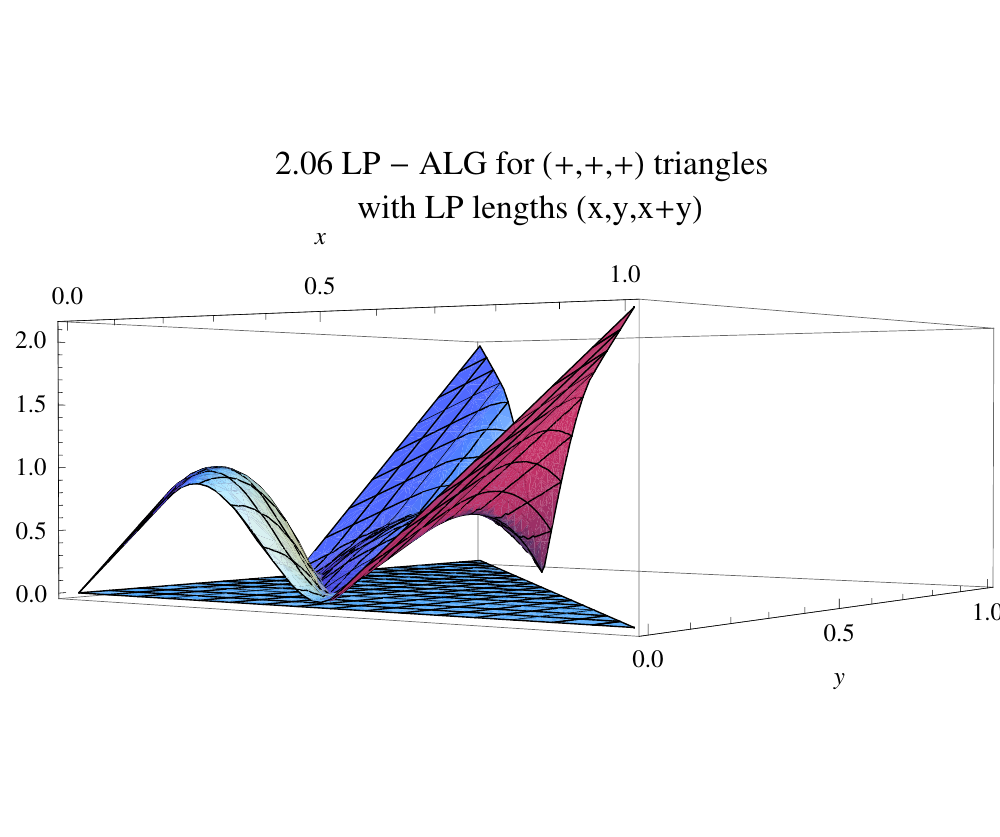}\hspace{1cm}
	\includegraphics[width=0.4\textwidth,natwidth=288,natheight=240]{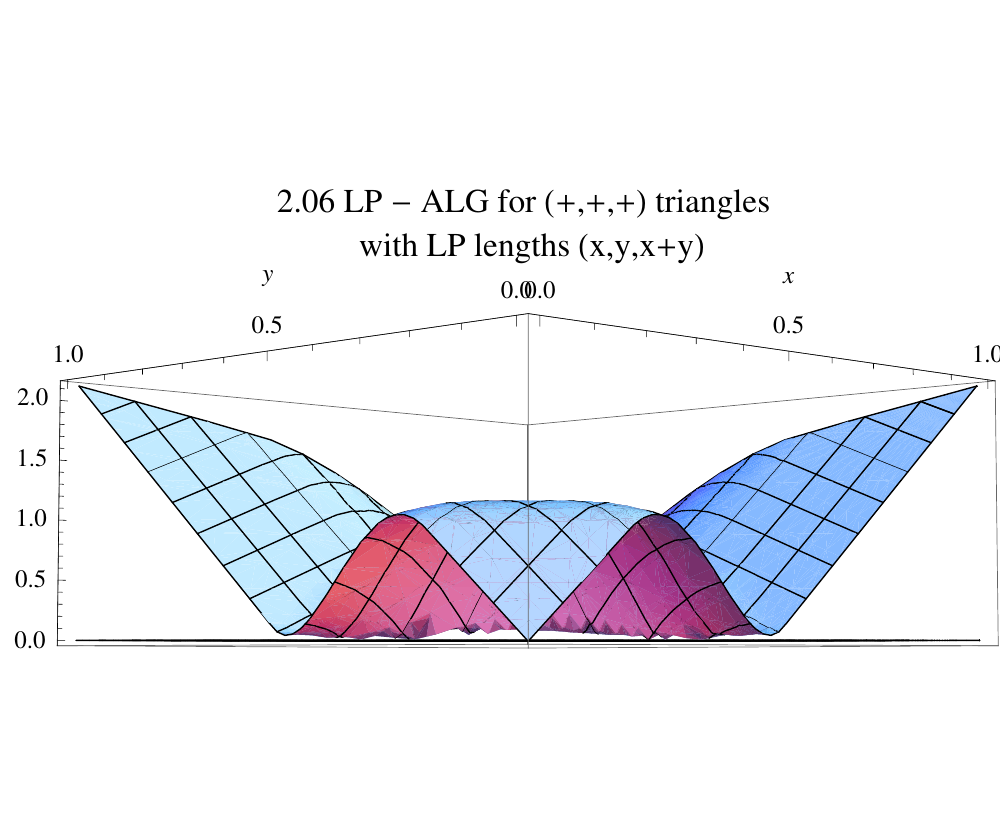}
	\caption{The difference between $\alpha \cdot LP(uvw) - ALG(uvw)$ for \ppp-triangles with tight triangle inequality constraints.}
	\label{fig:ppp_ratio}
\end{center}
\end{figure}
\begin{figure}[h!]
	\centering
	\includegraphics[width=0.4\textwidth,natwidth=288,natheight=240]{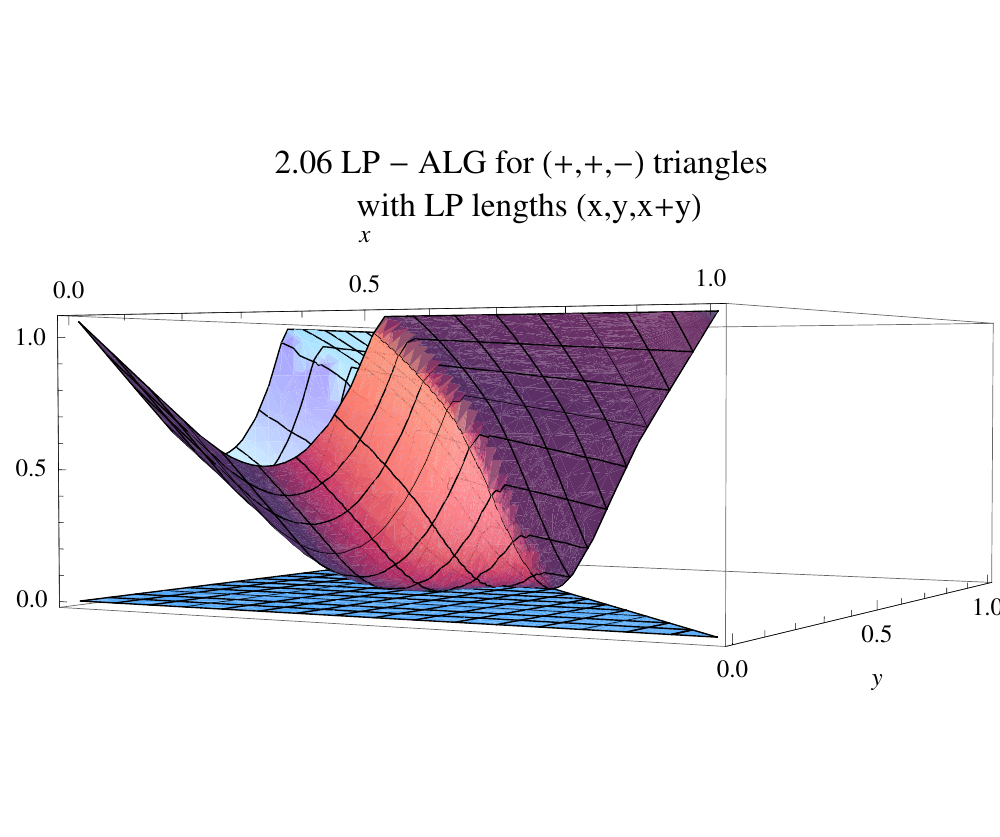} \hspace{1cm}
	\includegraphics[width=0.4\textwidth,natwidth=288,natheight=240]{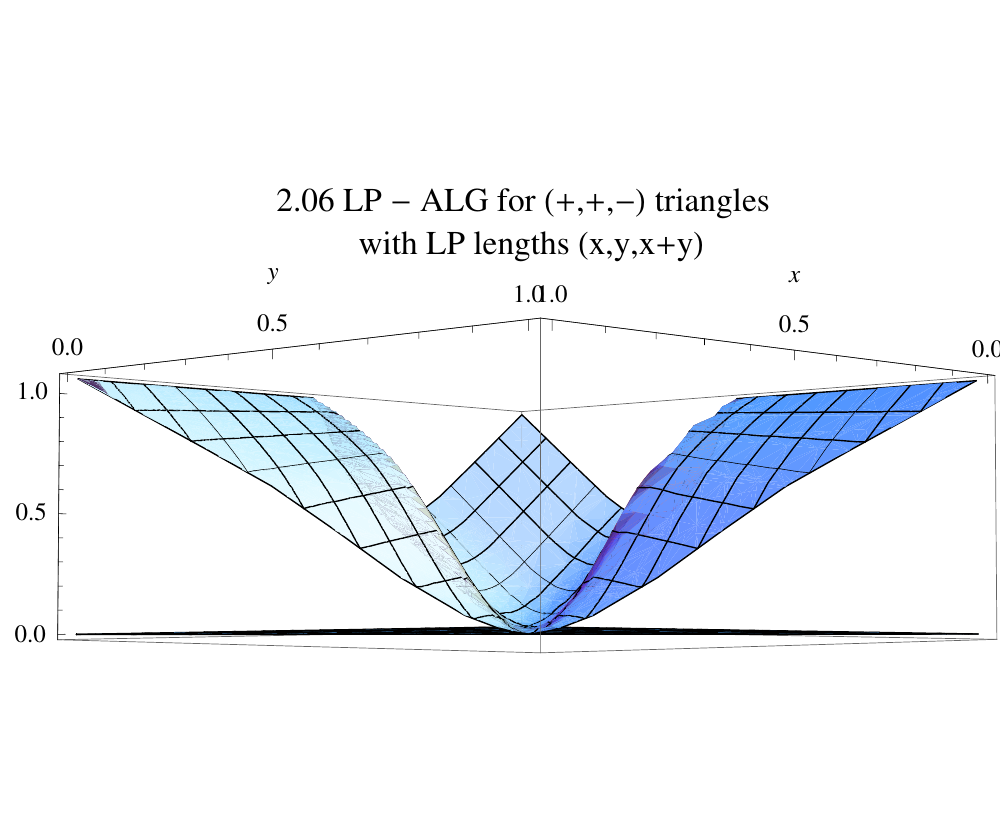}
	\caption{The difference between $\alpha\cdot LP(uvw) - ALG(uvw)$ for \ppm-triangles with tight triangle inequality constraints.}
	\label{fig:ppm_ratio}
\end{figure}
\begin{figure}[ht!]
	\centering
		\begin{minipage}{0.45\textwidth}
	\includegraphics[width=\textwidth,natwidth=288,natheight=240]{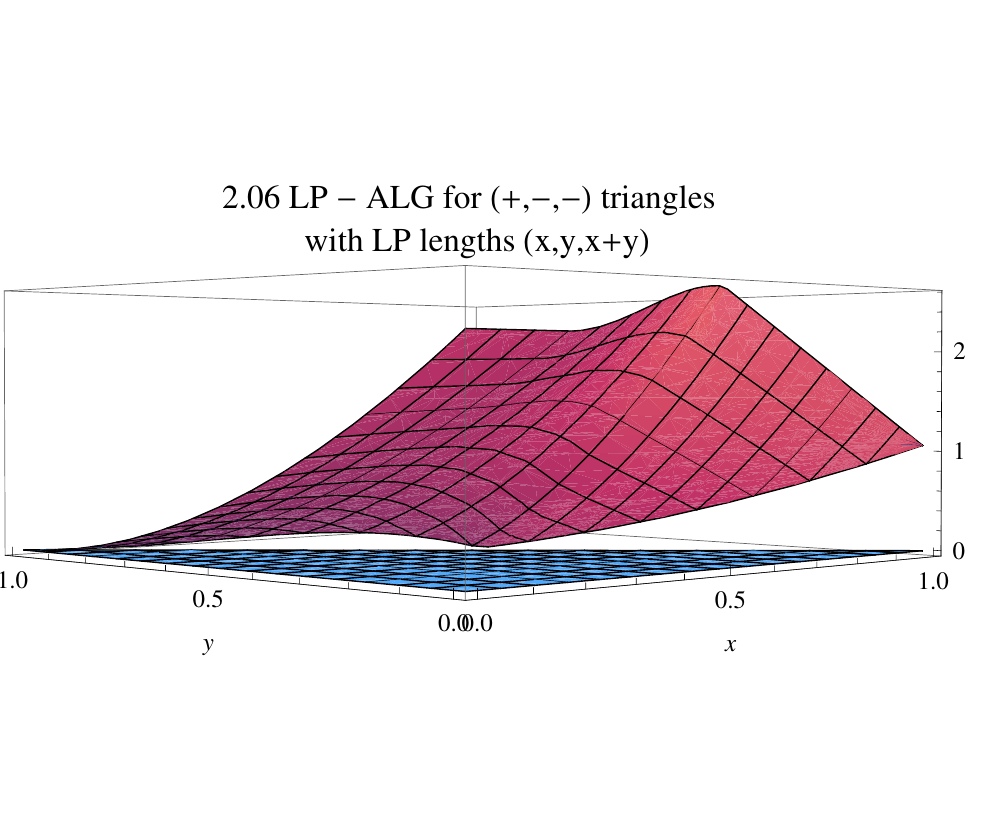}
	\begin{flushleft}
	\captionof{figure}{The difference between $\alpha\cdot LP(uvw) - ALG(uvw)$ for
	\pmm-triangles with tight triangle inequality constraints.}
	\label{fig:pmm_ratio}
	\end{flushleft}
	\end{minipage}\qquad
	\begin{minipage}{0.4\textwidth}
	\centering
	\includegraphics[width=\textwidth,natwidth=288,natheight=240]{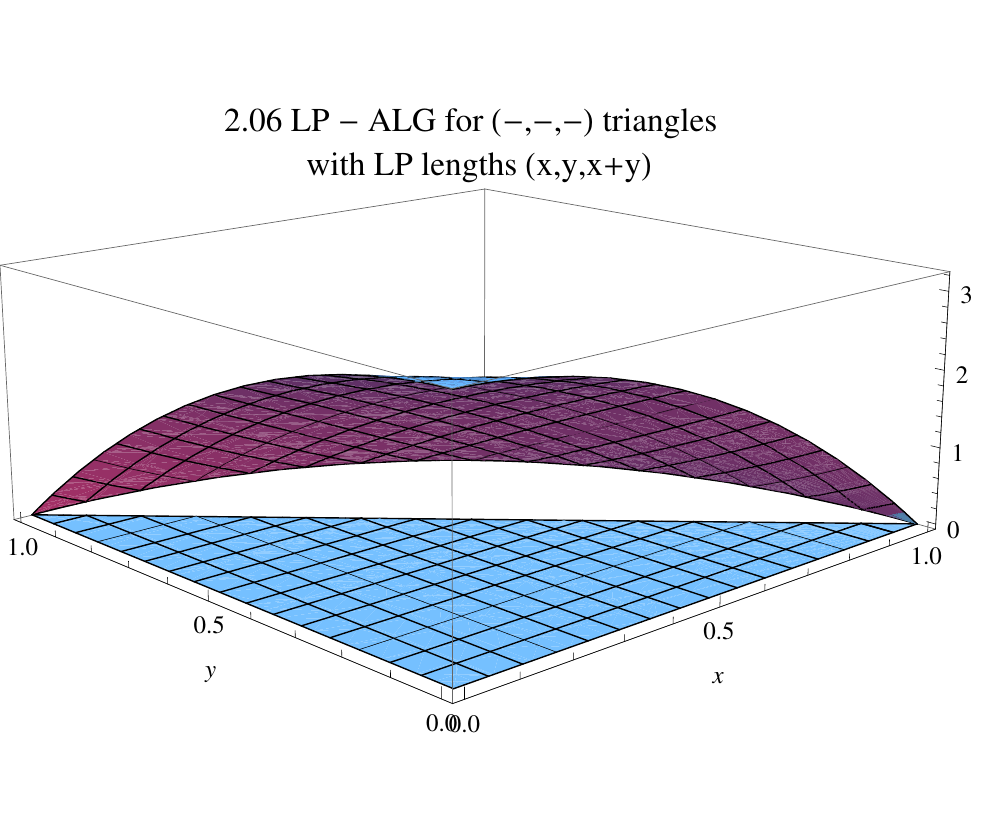}
	\begin{flushright}
	\captionof{figure}{The difference between $\alpha\cdot LP(uvw) - ALG(uvw)$ for \mmm-triangles with tight triangle inequality constraints.}
	\label{fig:mmm_ratio}
	\end{flushright}
	\end{minipage}
\end{figure}
\pagebreak

\section{Lower Bounds}\label{sec:gaps}
\subsection{Integrality Gap for Bipartite Graphs}
\begin{theorem}\label{thm:integrality-gap-bipartite}
For every constant $\delta > 0$, there exists an instance of bipartite correlation
clustering with an integrality gap of $3 - \delta$.
\end{theorem}

\begin{proof}
We use the following family of bipartite expanders.
\begin{theorem}[see e.g.,~\cite{LPS88}]\label{thm:bipartite-expander}
For every $d = p+1$ (where $p$ is prime), and for sufficiently large $n$ there exists a regular bipartite graph of
degree $d$ with girth $\Omega_d(\log n)$.
\end{theorem}

Let $G(V_1, V_2, E)$ be a bipartite expander as in
Theorem~\ref{thm:bipartite-expander}.  Let $V = V_1 \cup V_2$.  We define a
weight function $w \colon V_1 \times V_2 \rightarrow \{0,1\}$ as follows:
\begin{align*}
w(u,v) =
\begin{cases}
0, \text { if } (u,v) \in E \\
1, \text { otherwise}.
\end{cases}
\end{align*}
We define three sets of pairs:
\begin{align*}
E^+ & = E \\
E^- & = \{(u,v) | (u,v) \in V_1 \times V_2, (u,v) \notin E\}\\
E^{\emptyset} &= \{(u,v) | (u,v) \notin V_1 \times V_2\} .
\end{align*}
Consider the following solution to the linear program:
\begin{align*}
x_{uv} =
\begin{cases}
1/3, &\text{ if } e \in E^+  \\
1, &\text{ if }  e \in E^- \\
2/3, &\text{ if } e \in E^{\emptyset}.\\
\end{cases}
\end{align*}
Note that all triples of LP values satisfy triangle inequality except $(1/3, 1/3, 1)$-triples.
These latter triples don't exist by the bipartiteness of $G$.
Thus, we have a feasible LP solution of cost $|E|/3$.
Using Lemma~\ref{lem:clustering-cost-bipartite} below we can conclude that the ratio between the cost of the optimum solution and the best solution for the linear program is at least $3 (1 - 2/d)$. Taking large enough $d$ the proof or the theorem follows.
\end{proof}
It remains to show the following lemma:
\begin{lemma}\label{lem:clustering-cost-bipartite}
The cost of the optimum clustering for the weight function $w$ is at least $(1 - 2/d) |E|$.
\end{lemma}
\begin{proof}
Consider the optimum clustering of $G$ denoted $C_1, \dots, C_k$.
\begin{claim}\label{clm:cluster-size}
No cluster $C_i$ has size greater than $8d$.
\end{claim}
\begin{proof}
We give a proof by contradiction.
Assume there is a cluster $C_i$ of size greater than $8d$.
Let $s_1 = |C_i \cap V_1|$ and $s_2 = |C_i \cap V_2|$. W.l.o.g we can assume that $s_1 \ge s_2$.
We will show that splitting $C_i$ into two new clusters $C^1_i$ and $C^2_i$ gives a solution of smaller cost.
Let $C^1_i$ consist of $s_1 / 2$ arbitrary vertices from $C_i \cap V_1$ and $s_2 / 2$ arbitrary vertices from $C_i \cap V_2$. Let $C^2_i = C_i \setminus C^1_i$. The difference between the cost of the clustering $(C_1, \dots, C_i, \dots, C_k)$ and $(C_1, \dots, C^1_i, C^2_i, \dots, C_k)$ is equal to $|E^+ \cap \left(C^1_i \times C^2_i\right)| - |E^- \cap \left(C^1_i \times C^2_i\right)|$ or in other words the difference between the number of positive and negative edges between $C^1_i$ and $C^2_i$.
Note that:
\begin{align*}
|E^+ \cap \left(C^1_i \times C^2_i\right)| &\le s_2 d \\
 |E^- \cap \left(C^1_i \times C^2_i\right)| &\ge \frac{s_1 s_2}{2} - s_2 d.
\end{align*}
This implies that if $s_1 > 4d$ then the difference is negative, a contradiction.
\end{proof}
Using Claim~\ref{clm:cluster-size} and the fact that the girth of $G$ is $\Omega(\log |V|) > 8d$ we conclude that the subset of $E^+$ contained inside each $C_i$ forms a forest.
Thus, the total number of edges from $E^+$ which are not contained inside any cluster is at least $|E^+| - \sum_{i = 1}^k (|C_i| - 1) \ge |E^+| - n = (1 - 2/d) |E^+|$.
\end{proof}

\subsection{Integrality Gap for Complete Graphs with Triangle Inequalities}
See Section~\ref{sec:triangle-inequalities} for a formal definition of correlation clustering
in weighted complete graphs with triangle inequalities.
\begin{theorem}\label{thm:integrality-gap-triangle-inequalities}
For every $\delta > 0$ there exists an instance of correlation clustering in complete graphs, which satisfies triangle inequalities and has integrality gap $6/5 - \delta$.
\end{theorem}
\begin{proof}
Let $V = V_1 \cup V_2$ where $|V_1| = |V| / 2 = n$.
We define a weight function $w \colon V \times V \rightarrow [0,1]$ as follows:
\begin{align*}
w(u,v) =
\begin{cases}
1/3, \text{ if } (u,v) \in V_1 \times V_2 \\
2/3, \text{ otherwise}.
\end{cases}
\end{align*}
Consider the following solution to the linear program:
\begin{align*}
x_{uv} =
\begin{cases}
1/2,  \text{ if } (u,v) \in V_1 \times V_2  \\
1, \text{ otherwise.}
\end{cases}
\end{align*}
This solution satisfies triangle inequalities. The cost of each edge $(u,v) \in V_1 \times V_2$ is $1/2$ and the cost of each edge $(u,v) \notin V_1 \times V_2$ is $1/3$. Overall, the cost is $\frac12 n^2 + \frac13 \cdot 2 \binom{n}{2} = \frac56 n^2 + O(n)$.
As we will show below in Lemma~\ref{lem:clustering-cost-triangle-inequalities}, the cost of the optimum clustering is at least $n^2 - O(n)$ concluding the proof.
\end{proof}

\begin{lemma}\label{lem:clustering-cost-triangle-inequalities}
The cost of the optimum clustering for the weight function $w$ is at least $n^2 - O(n)$.
\end{lemma}
\begin{proof}
We will show that the optimum clustering is a single cluster $V$. Then the lemma follows since the cost of such clustering is $\frac13 n^2$ for the edges in $V_1 \times V_2$ plus $\frac23 \cdot 2 \binom{n}{2}$ for the edges not in $V_1 \times V_2$, giving the overall cost of $n^2 - O(n)$.

We will give a proof by contradiction, assuming that the optimum clustering has multiple clusters $C_1, \dots, C_k$.
For every cluster $C_i$ let $s_i = |V_1 \cap C_i|$ and $t_i = V_2 \cap C_i$.
Because $|V_1| = |V_2|$ there exists a pair of clusters $C_i$ and $C_j$ such that $s_i \ge t_i$ while $s_j \le t_j$.
Consider a new clustering where $C_i$ and $C_j$ are replaced with $C_i \cup C_j$, while all other clusters are unchanged.
The difference between the cost of the new clustering and $C_1, \dots, C_k$ is equal to:
\begin{align*}
\frac13(s_i s_j + t_i t_j) - \frac13 (s_i t_j + s_j t_i) = \frac13 (s_j - t_j)(s_i - t_i) \le 0.
\end{align*}

\end{proof}

\subsection{\texorpdfstring{Lower Bound of $2.025$ for the Complete Case}{Lower Bound of 2.025 for the Complete Case}}\label{sec:lbd}
We will now argue that within the algorithm and analysis framework described in Section~\ref{sec:alg}, no choice of the rounding functions $\fp$ and $\fm$ can give an approximation factor better than $2.025$. More specifically, no functions $f^+$ and $f^-$
satisfy the conditions of Lemma~\ref{lem:alpha} with $\alpha = 2.025$.

\begin{theorem}\label{thm:limitation-complete}
For any functions $\fp$ and $\fm$ that satisfy the conditions of
Lemma~\ref{lem:alpha} on all triangle types, the ratio $\alpha$ must be larger than $2.025$.
\end{theorem}
\begin{proof}
Our lower bound argument follows the same kind of approach we used in Section~\ref{sec:choosing_f} to identify good rounding functions for the upper bound, however, presents a better analysis of the \ppm-triangles with edge lengths $(x,x,2x)$.

Fix $\fp$ and $\fm$ satisfying the conditions of Lemma~\ref{lem:alpha} on all triangle types
for some $\alpha$. From \prettyref{lem:fp_upper} and \prettyref{cor:fmbound}, we know that
\begin{align}
f^+(x)&\le 1 - \sqrt{1 - \alpha x};\label{eq:fp-bound}\\
f^-(2x)&\ge \sqrt{1-\alpha (1-2x)}.\label{eq:fn-bound}
\end{align}
Using equations \prettyref{eq:ecost} and  \prettyref{eq:elp}
for \ppm-triangles with edge lengths $(x,x,2x)$ (where $x\in[0,1/\alpha]$),
we write $\alpha LP -ALG \geq 0$ as follows:
$$\alpha \bigl(\underbrace{1-f^+(x)^2 +2x f^+(x)^2 - 2xf^+(x)f^-(2x)}_{LP}\bigr) -
\bigl(\underbrace{1 + f^+(x)^2 +2f^-(2x)  -4f^+(x)f^-(2x)}_{ALG}\bigr) \geq 0.$$
Rearranging terms, we get
\begin{equation}\label{eq:sq2}
(1 + \alpha - 2\alpha x) f^+(x)^2 + 2(1 - f^+(x) (2- \alpha x))f^-(2x) \leq \alpha  - 1.
\end{equation}
We claim that the coefficient $2(1 - f^+(x) (2- \alpha x))$ of $f^-(2x)$ is always
positive: From~(\ref{eq:fp-bound}) using $(2-\alpha x)\geq 0$, we get
$$1 - f^+(x) (2- \alpha x) \geq 1 - (1 -\sqrt{1 - \alpha x}) (2- \alpha x)=
\sqrt{1 - \alpha x} \cdot ((2 - \alpha x) - \sqrt{1-\alpha x})\geq 0.$$
The last inequality holds, since $(2 - \alpha x)\geq 1 \geq \sqrt{1-\alpha x}$.
We now replace $f^-(2x)$ in (\ref{eq:sq2}) with $\sqrt{1-\alpha (1-2x)}$.
By inequality~(\ref{eq:fn-bound}), we have
$$(1 + \alpha - 2\alpha x) f^+(x)^2 + (2 - 2f^+(x) (2- \alpha x)) \cdot \sqrt{1-\alpha (1-2x)} \leq \alpha  - 1.$$
Once again rearranging terms we get
\begin{align}
	\label{eq:root_bounds}
	(1 + \alpha - 2\alpha x) f^+(x)^2 - 2\bigl((2- \alpha x) \cdot \sqrt{1-\alpha (1-2x)}\bigr) f^+(x)  + 2\sqrt{1-\alpha (1-2x)} - \alpha  + 1 &\leq 0.
\end{align}
This is a quadratic equation with a positive leading coefficient. The left hand side
is negative if $f^+(x)$ lies between the roots of the equation.  For $x = 0.48$ and $\alpha = 2.025$, we get $f^+(x)\in (0.836, 0.987)$ (the first root is rounded down; the second root is rounded up). This contradicts to inequality~(\ref{eq:fp-bound}), since $1 - \sqrt{1 - \alpha x} < 0.833$ (See \prettyref{fig:lb}).
We conclude that $\alpha = 2.025$ does not satisfy the conditions of Lemma~\ref{lem:alpha}.
\end{proof}

\begin{figure}
\centering
\includegraphics[width=0.5\textwidth,natwidth=288,natheight=212]{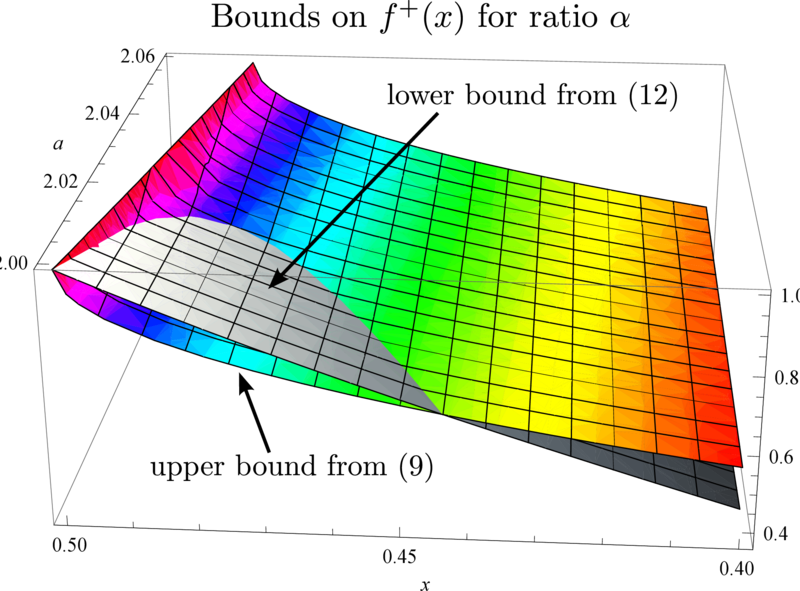}
\caption{Plot of the bounds on $f^+(x)$ as a function of the
  approximation ration $\alpha$; the upper bound from \prettyref{eq:fp-bound} is shown in rainbow, and the lower bound from \prettyref{eq:root_bounds} is shown in gray. The region of intersection is where we obtain the contradiction.}
\label{fig:lb}
\end{figure}

\subsection{\texorpdfstring{Lower bound of $1.5$ for Weighted Triangle Inequalities}{Lower bound of 1.5 for Weighted Triangle Inequalities}}
A similar bound can be obtained for the weighted triangle inequalities
case.

\begin{theorem}\label{thm:limitation-wti}
	For any monotone function $f(x_{uv},\lambda_{uv})$, the approximation ratio $\alpha$ must be at least $1.5$.
\end{theorem}

The proof proceeds in a manner similar to that of
\prettyref{thm:limitation-complete}; we omit some of the details, but give a
proof sketch below.
\begin{proof}(Sketch)
Let $f(w,\ell)$ be the probability of cutting an edge with edge weight $w$ and LP
value $\ell$.

Let $\mathcal{C}(w_1,w_2,w_3,\ell_1,\ell_2,\ell_3)$ be the polynomial given by
$\alpha\cdot LP - ALG$ for a triangle with edge weights $w_1,w_2,w_3$ and
corresponding LP lengths $\ell_1,\ell_2,\ell_3$ (note that the costs are
calculated differently in the case of weighted edges, as compared to the
integral edge constraints case). We prove our lower bound by showing that there
is no $f(x,x)$ that gets an approximation $\alpha < 1.5$; that is, for the
purpose of this proof we will only take $\ell_1 = w_1, \ell_2 = w_2$, and
$\ell_3 = w_3$, and so we abbreviate $\mathcal{C}(w_1,w_2,w_3)$ for ease of
notation.

We first consider a triangle with weights $(0,x,x)$. Clearly, it must be that
$f(0,0) = 0$. Thus, we can solve for the roots of $f(x,x)$ in
$\mathcal{C}(0,x,x)$, which has leading positive coefficient, and so we
have that to get approximation $\alpha$,
\begin{align*}
	f(x,x) &\le 1 - x - \sqrt{1- x(1+2\alpha) + x^2(1 + 2\alpha)},\\
\intertext{\begin{center}or\end{center}}
	f(x,x) &\ge 1 - x + \sqrt{1- x(1+2\alpha) + x^2(1 + 2\alpha)}.
\end{align*}
for $\alpha < \frac{3}{2}$ there are values of $x \in [0,1]$ for which these roots are complex, and in particular this means that there must be a discontinuity in $f(x,x)$; the figure below shows the area in which $f(x,x)$ may be to achieve $\alpha = 1.499$.
\begin{center}
	\includegraphics[width=0.4\textwidth]{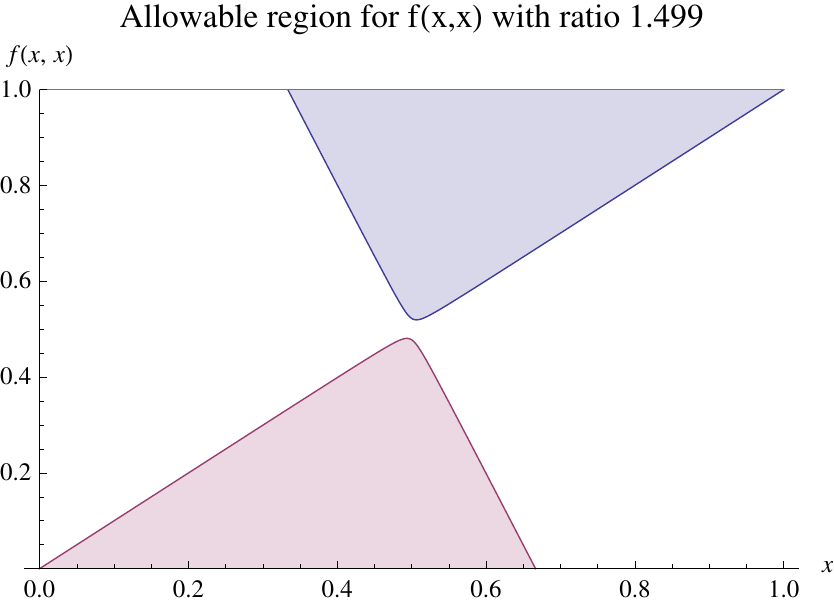}
\end{center}

It thus remains to rule out a discontinuous $f(x,x)$. Thus, suppose to the
contrary that there is a discontinuity in $f(x,x)$ at $x = y$. By considering a
triangle that has one edge just barely longer than $y$ and another edge just
barely shorter than $y$, we will obtain our contradiction.  The minimum value
of the upper root is larger than $\tfrac{1}{2}$, and the maximum value of the
lower root is less than $\tfrac{1}{2}$, therefore as $\epsilon \to 0$,
$f(y+\epsilon, y+\epsilon) = \tfrac{1}{2} + c_1$ for some fixed positive $c_1$,
and similarly $f(y-\epsilon, y-\epsilon) = \tfrac{1}{2} - c_2$ for some fixed
positive $c_2$. By examining the roots of $y$ in $\mathcal{C}(2\epsilon,
y-\epsilon, y +\epsilon)$ as $\epsilon \to 0$, we see that there are no values
of $c_1,c_2$ for which $f(x,x)$ gives an approximation better than $\alpha =
\frac{3}{2}$, and we have our conclusion.
\end{proof}

\section{Weighted Case}\label{sec:weighted}
In this section, we consider weighted instances of the problem. We assume that every
edge is partially positive and partially negative. Specifically, every edge $(u,v)$
has a positive weight $\lambda^+_{uv}$ and a negative weight $\lambda^-_{uv}$. The goal now
is to minimize the weight of unsatisfied constraints defined as follows:
$$\sum_{(u,v)\in E} \lambda^+_{uv}\I((u,v) \text{ is cut}) + \lambda^-_{uv}\I((u,v) \text{ is not cut}).$$
We study complete graphs with $\lambda^+_{uv} + \lambda^-_{uv} = 1$.

This problem has been studied by Charikar et. al.~\cite{CGW05}, Ailon et.
al.~\cite{ACN08}, and others. The edge weights represent our confidence level
that the edge is positive or negative.

A slight modification of our algorithm gives $2.06$ for this problem: For every edge $(u,v)$ the algorithm flips a coin and assigns $p_{uv} = f^+(x_{uv})$ with probability $\lambda^+_{uv}$ and $p_{uv} = f^-(x_{uv})$
with probability $\lambda^-_{uv} = 1 - \lambda^+_{uv}$. The analysis of the
algorithm are essentially the same as for unweighted case, since in the
weighted case, the expression $\calC = \alpha LP - ALG$ can be represented
as a convex combination of unweighted expressions $\calC = \alpha LP - ALG$.

Here we give a black box reduction of the weighted problem to the unweighted problem.
The reduction works with any approximation algorithm.

\begin{theorem}
If there exists an $\alpha$-approximation algorithm for the unweighted Correlation Clustering Problem on complete graphs, then there exists an $\alpha$-approximation algorithm
for the Weighted Correlation Clustering problem on complete graphs.
\end{theorem}
\begin{proof}[Sketch of the proof]
Given a weighted instance of the problem, we construct an unweighted instance:
We pick a sufficiently large integer number $N$ e.g. $N=100n^2$. Then, we replace
every vertex $u$ with $N$ vertices $u_i$, $i\in\{1,\dots, N\}$. We add a positive edge between
$u_i$ and $v_j$ with probability $\lambda^+_{uv}$ and a negative edge with
probability $\lambda^-_{uv}$. Then, we add positive edges between all $u_i$ and $u_j$.
Thus, we get a complete unweighted graph $G'$. We run the $\alpha$-approximation algorithm
on $G'$ and obtain a clustering of vertices of $G'$. Now, for every $u\in V$,
we pick a random $i\in\{1,\dots, N\}$ and assign $u$ to the cluster of $u_i$.

It is easy to see that the algorithm gives $\alpha$ approximation, since
(1) the cost of the optimal solution for $G'$ is at most $N^2$ times the cost of the optimal
solution for $G$; (2) the expected cost of clustering for $G$  equals to the cost of
the solution for $G'$ scaled by $1/N^2$ up to
an additive error term of $\varepsilon = 1/10$. The latter statement follows from the following lemma.

\begin{lemma}
With high probability over choice of random labels $s_{ij}\in \{\pm\}$ on the edges $(u_i,v_j)$,
for all clusterings of vertices $u_i$ and $v_j$ the following inequality holds:
\begin{multline}\label{eq:weighted-chernoff}
w^+_{uv}\Pr_{i,j}((u_i,v_j)\text{ is cut}) + w^-_{uv}\Pr_{i,j}((u_i,v_j)\text{ is not cut})
\leq \\ \leq
\frac{1}{N^2} \sum_{i,j: s_{ij}=``+"} \I ((u_i,v_j)\text{ is cut})
+
\frac{1}{N^2} \sum_{i,j: s_{ij}=``-"} \I ((u_i,v_j)\text{ is not cut})
+\varepsilon/n^2.
\end{multline}
Here, probability is taken over random $i,j\in \{1,\dots, N^2\}$.
\end{lemma}
\begin{proof}
The proof is very standard: There are at most $\exp(2N\ln (2N))$ ways to partition
vertices $u_i$ and $v_j$. For a every fixed partition, the probability that
(\ref{eq:weighted-chernoff}) is violated is at most $\exp(-\varepsilon^2 N^2/2)$ (by the
Chernoff bound; note, that the number of pairs $(i,j)$ is $N^2$). Hence, by the union bound,
(\ref{eq:weighted-chernoff}) holds for all partitions with high probability.
\end{proof}
\end{proof}

\subsection{Weights Satisfying the Triangle Inequality}\label{sec:triangle-inequalities}
Gionis, Mannila, and Tsaparas~\cite{GMT07} proposed a variant of the Weighted
Correlation Clustering problems with negative weights satisfying the triangle
inequality i.e., $$\lambda^{-}_{uw} \leq \lambda^-_{uv} + \lambda^{-}_{vw},$$
for all $u,v,w\in V$. They gave a 3-approximation algorithm for the problem.
Ailon, Charikar, and Newman~\cite{ACN08} improved the approximation factor to $2$.
We show that for the rounding functions $f^+(x) = x^2$ and $f^-(x)= \sqrt{x}$,
the approximation ratio is $1.53$. Our proof is computer assisted.  Note that
$x^2$ is a convex function; and $\sqrt{x}$ is a concave function, so we can
apply~Lemma~\ref{lem:tight-triangle} and limit the case analysis only to
triangles with tight triangle inequality constraints. For an appropriate $c$,
if we pick $f^+(x) = \min((4 - 2 \sqrt{2})x^2, 1)$ and $f^-(x)= \sqrt{x}$, we
get a $1.5$ approximation.
\begin{theorem}\label{thm:weighted-triang-ineq}
There exists a deterministic polynomial-time algorithm for the Correlation
Clustering Problem that gives a $1.5$-approximation for weighted complete
graphs satisfying the triangle inequalities on $\lambda^{-}_{uv}$.
\end{theorem}
The proof is also computer assisted.

\section{Derandomization}\label{sec:derand}
We now derandomize the approximation algorithm presented in the previous section. The deterministic algorithm
will always return a solution of cost at most $\alpha LP$, where $\alpha$ is the parameter from Lemma~\ref{lem:alpha}.

\begin{theorem}
Suppose that the conditions of Lemma~\ref{lem:alpha} are satisfied for a set of functions $\{\fpos, \fneg, \fno\}$
and a parameter $\alpha$. Then there exists a polynomial-time deterministic algorithm with approximation ratio
$\alpha$.
\end{theorem}
\begin{proof}
We slightly tweak the algorithm: After the algorithm set probabilities $p_{uv}$ using functions $f$,
we change the value of each $p_{uv}$ to either $0$ or $1$ to maximize the function:
\begin{equation}\label{eq:sum-uvw}
\sum_{u,v,w \in V_t} \alpha LP(uvw) - ALG(uvw),
\end{equation}
which is equal to $(\alpha \Exp[LP_t|V_t] - \Exp[ALG_t|V_t])\times 6|V_t|$(see the previous section for details).
For each $p_{uv}$, the algorithm fixes the value of all other variables $p_{u'v'}$ and picks the value for $p_{uv}$ greedily.
Then, the algorithm picks the pivot $w$ for which $\alpha \Exp[LP_t|V_t, w_t=w] - \Exp[ALG_t|V_t,w_t =w]$
is maximized. Then, the algorithm as before adds every $u$ to $S_t$ with probability $(1-p_{uw})$.
Now, however, all probabilities $p_{uv}$ are equal to $0$ or $1$. So the algorithm assigns each $u$ to $S_t$ deterministically.

The main observation is that function~(\ref{eq:sum-uvw}) is convex in every variable $p_{uv}$. Indeed, all terms
$LP(uvw)$ and $ALG(uvw)$ are quadratic functions of $p_{uv}$, $p_{vw}$, and $p_{uw}$. Moreover,
all terms $\ecost_w(u,v)$ and $\lp_w(u,v)$ (except for $\ecost_w(u,u)$) are linear in each variable $p_{uv}$. The only nonlinear terms are
$\ecost_w(u,u)=2(1-p_{uw})p_{uw}$ for $u$ having a self-loop. But it is concave as a function of $p_{uw}$. So $\alpha LP(uvw) - ALG(uvw)$ is a convex function in every variable
$p_{uv}$.

Therefore, after rounding each $p_{uv}$ to 0 or 1 the value of~(\ref{eq:sum-uvw}) may only increase. Initially, after setting $p$'s using functions $f$, it is nonnegative. So after replacing all $p$'s with $0$'s and $1$'s, it is also nonnegative. Hence, $\alpha \Exp[LP_t|V_t] - \Exp[ALG_t|V_t]\geq 0$, and,
consequently, for some pivot $w$, $\alpha \Exp[LP_t|V_t, w_t=w] - \Exp[ALG_t|V_t,w_t=w]\geq 0$. We showed that for some pivot $w$, the number
of violated constraints at step $t$ is upper bounded by $\alpha$ times the LP volume removed at this step. Hence, the total number of violated constraints
is upper bounded by $\alpha LP$.
\end{proof}

\section{\texorpdfstring{Analysis of the $k$-partite case}{Analysis of the k-partite case}}\label{sec:kpart}
\begin{proof}[Proof of Theorem~\ref{thm:kpart}]
By Lemma~\ref{lem:alpha}, it suffices to show that $ALG (uvw)\leq LP (uvw)$ for
every triangle~$uvw$ in the graph. There are $7$ possible types of triangles in a complete $k$-partite graph:
\ppp-triangles, \ppm-triangles, \pmm-triangles, and \mmm-triangles may exist
when each vertex is a member of different partitions;
\ppn-triangles, \pmn-triangles, and \mmn-triangles can occur when two vertices
are members of the same partition and the third vertex is a member of a
different partition. When all three vertices belong to one partition, all edges
are neutral, and so none contribute to the cost.

We adopt the convention that the triangles have LP lengths $(a,b,c)$, with vertex
$u$ opposite edge $b$, vertex $v$ opposite edge $c$, and vertex $w$ opposite
edge $a$ (see figure below). Costs incurred when vertex $x \in \{u,v,w\}$ is a pivot will be enclosed in square brackets with subscript $x$, as in $[cost]_x$.
Instead of writing, $ALG(uvw)$ and $LP(uvw)$, we write $ALG$ and $LP$.
		\begin{center}
			\tikz\draw (1,1.3) node[above]{u} -- node[right]{a}
			(2,0) node[right]{v} -- node[below]{b}
			(0,0) node[left]{w} -- node[left]{c} (1,1.3);
		\end{center}

\medskip

\noindent\textbf{\pmn-Triangles:}
We first analyze the most interesting case: what happens if our triangle has edge labels \pmn. Consider a \pmn-triangle with side lengths
$(a,b,c)$ respectively. We have
\begin{align*}
ALG &= [(1-p_{uv})(1-p_{uw})]_u + [(1-p_{uw})p_{vw} + (1-p_{vw})p_{uw}]_w\\
    &= 1 - f^+_3(a) + f^-_3(b) + f^+_3(a)\fnk(c) - 2f^-_3(b)\fnk(c);\\[5pt]
LP  &= [(1-b)\cdot(1-p_{uw}p_{uv})]_u + [a\cdot(1-p_{uw}p_{vw})]_w\\
&= (1-b)(1-f^+_3(a)\fnk(c)) + a(1 - f^-_3(b)\fnk(c)).
\end{align*}
Since $f^+_3$ and $\fnk$ are piecewise functions, we consider four cases.

\medskip

\noindent 1. If $a < \nicefrac{1}{3}$, $c < \nicefrac{2}{3}$, then $LP = 1-b + a - \tfrac{3}{2}abc$, and $ALG = 1 + b - 3bc$. So,
\begin{align*}
3\cdot LP - ALG &= 2 + 3a + b(3c - 4 - \frac{9}{2}ac)\ge 2 + 3a + (a+c)(3c - 4 - \frac{9}{2}ac)\\
				&= (2 - a) + ac(3 - \frac{9}{2}a - \frac{9}{2}c) + (3c^2 - 4c).
\end{align*}
The inequality above follows from the triangle inequality constraints $b\leq a + c$. We now bound each term separately using the
assumptions $a<1/3$ and $c<2/3$: We have $(2-a) > 5/3$, $(3 - 9a/2 - 9c/2)\geq -3/2$, and $(3c^2-4c)\geq -4/3$. Thus,
$$3\cdot LP - ALG \ge \tfrac{1}{3} - \tfrac{3}{2}ac \ge 0.$$
as desired.

\medskip

\noindent 2. If $a < \tfrac{1}{3}$, $c \ge \tfrac{2}{3}$, then $LP = 1 - b + a - ab$, and
$ALG = 1 - b$. Because $a - ab > 0$, $ALG \le LP$, as desired.

\medskip

\noindent 3. If $a \ge \tfrac{1}{3}$, $c < \tfrac{2}{3}$, then $LP = 1-b -\tfrac{3}{2}c + \tfrac{3}{2}bc + a -\tfrac{3}{2}abc$ and $ALG = b + \tfrac{3}{2}c - 3bc$. We have
\begin{align*}
3\cdot LP - ALG &= 3 + 3a - 4(b+c) + \tfrac{15}{2}bc - \tfrac{9}{2}abc
				\ge 4 - 4(b+c) + \tfrac{15}{2}bc - \tfrac{3}{2}bc\\
				&= 4(1-b)(1-c),
\end{align*}
where in the second inequality we applied $a \ge 1/3$. Thus, $ALG \le 3\cdot LP$.

\medskip

\noindent 4. If $a \ge \tfrac{1}{3}$, $c \ge \tfrac{2}{3}$, then $LP = 0$ and $ALG = 0$ as well, so $ALG \le 3\cdot LP$.

We now bound the expected costs of the $6$ other triangles.

\medskip

\noindent\textbf{\ppp-Triangles:} If $a,b,c <\tfrac{1}{3}$ or $a,b,c \ge \tfrac{1}{3}$, then
either no edges are cut and the algorithm makes no mistakes, or all edges are cut and
no edge is ever charged unsafely, so $ALG = 0$. The remaining cases are
$a,b < \tfrac{1}{3} \le c$ and $a < \tfrac{1}{3} \le b,c$. In both cases,
			\begin{align*}
				ALG &= [f^+_3(c)(1-f^+_3(a)) + f^+_3(a)(1-f^+_3(c))]_u
					 + [f^+_3(a)(1-f^+_3(b)) + f^+_3(b)(1-f^+_3(a))]_v\\
					 &\qquad + [f^+_3(b)(1-f^+_3(c)) + f^+_3(c)(1-f^+_3(b))]_w,\\
				LP &= [b(1-f^+_3(c)f^+_3(a))]_u
					 + [c(1-f^+_3(a)f^+_3(b))]_v
					 + [a(1-f^+_3(b)f^+_3(c))]_w.
				\end{align*}
			When $a,b < \tfrac{1}{3} \le c$, $ALG = 2$, $LP = a+b+c \ge 2c \ge
			\tfrac{2}{3}$, where we apply the LP triangle inequality
			constraints. When $a < \tfrac{1}{3} \le b,c$, $ALG = 2$, $LP = c +
			b \ge \tfrac{2}{3}$. Therefore $3\cdot LP \ge ALG$ as desired.

\medskip

\noindent\textbf{\mmm-Triangles:} In this case, the function $f^-_3$ is the only one participating in the costs; the costs are
		\begin{align*}
			ALG &= [(1-f^-_3(a))(1-f^-_3(c))]_u + [(1-f^-_3(a))(1-f^-_3(b))]_v + [(1-f^-_3(c))(1-f^-_3(b))]_w\\
				&= 3 - 2a - 2b - 2c + ab + bc + ac,\\
			LP &= [(1-b)(1-f^-_3(a)f^-_3(c))]_u + [(1-c)(1-f^-_3(a)f^-_3(b))]_v + [(1-a)(1-f^-_3(c))f^-_3(b))]_w\\
			   &= 3 - a - b - c - ab - ac - bc + 3abc.
		\end{align*}
		We verify that the $3$-approximation holds:
		\begin{align*}
			3\cdot LP - ALG &= 6 - a - b - c - 4ab - 4ac - 4bc + 9abc \\
							&= 6 - a - b - 4ab - c(4a(1-b) + 4b(1-a) + (1 - ab))\\
							&\ge 5 - 5a - 5b + 5ab= 5(1-a)(1-b),
		\end{align*}
		where in the third inequality we have applied $c \le 1$. Thus, $3\cdot LP \ge ALG$ as desired.

\medskip

\noindent\textbf{\ppm-Triangles} Here, we must verify the cases $a,b < \tfrac{1}{3}$,
		$a < \tfrac{1}{3} \le b$, and $\tfrac{1}{3} \le a,b$. In all cases,
	\begin{align*}
		ALG &= [\fpk(a)(1-\fmk(c)) + \fmk(c)(1-\fpk(a))]_u
		+ [(1-\fpk(a))(1-\fpk(b))]_v
			\\&\qquad  + [f^+_3(b)(1-f^-_3(c)) + f^-_3(c)(1-f^+_3(b))]_w,\\
		LP &= [b(1 - \fpk(a)\fmk(c))]_u + [(1-c)(1-\fpk(a)\fpk(b))]_v + [a(1-\fpk(b)\fmk(c))]_w.
	\end{align*}
	When $a,b < \tfrac{1}{3}$, we have $ALG = 1 + 2c$, $LP = 1 - c + b + a\ge 1$, where we have applied the triangle inequality constraint to bound the LP. Since $c \le 1$, we have $3\cdot LP \ge ALG$.

	When $a < \tfrac{1}{3} \le b$, we have $ALG = 1$, $LP = 1 + a + b -c -ac\ge
	1 - ac \ge \frac{2}{3}$, where we have applied the triangle inequality
	constraint and the fact that $ac \le \tfrac{1}{3}$.

	When $\tfrac{1}{3} \le a,b$, we have $ALG = 2(1 -c)$, $LP = b(1-c) + a(1-c)
	\ge \tfrac{2}{3}(1-c)$, where we have applied the fact that $a+b \ge
	\tfrac{2}{3}$.  Thus, we have $3\cdot LP \ge ALG$ in this case as well.

\medskip

\noindent\textbf{\pmm-Triangles:} Here we must verify two cases: $a < \tfrac{1}{3}$ and $a \ge \tfrac{1}{3}$. In both cases, the costs are
	\begin{align*}
		ALG &= [(1-\fpk(a))(1-\fmk(c))]_u
		      + [(1-\fpk(a))(1-\fmk(b))]_v\\
			  &\qquad + [\fmk(c)(1-\fmk(b)) + \fmk(b)(1-\fmk(c))]_w,\\
		LP &= [(1-b)(1-\fpk(a)\fmk(c))]_u +[(1-c)(1-\fpk(a)\fmk(b))]_v + [a(1-\fmk(b)\fmk(c))]_w.
	\end{align*}
	When $a < \tfrac{1}{3}$, $ALG = 2 - 2bc$, $LP = 2 -b -c + a - abc$. We calculate,
	\begin{align*}
		3\cdot LP - ALG &= 4 - 3b - 3c + 3a + 2bc - 3abc
		= 4(1-b)(1-c) + b + c - 2bc + 3a(1-bc)\\
		&= (1 - bc) + 3(1-b)(1-c) + 3a(1-bc)
		\ge 0,
	\end{align*}
	as desired.

	When $a \ge \tfrac{1}{3}$, $ALG = c + b -2bc$, $LP = 2 - 2b - 2c + 2bc + a(1-bc)$. Again we verify,
	\begin{align*}
		3\cdot LP - ALG &= 6 - 7b - 7c + 8bc + 3a(1-bc) \ge 7 - 7b - 7c + 7bc \\
			   &= 7(1-b)(1-c) \ge 0,
	\end{align*}
	where to get the inequality we applied the assumption $a \ge \frac{1}{3}$.

\medskip

\noindent\textbf{\ppn-Triangles:} Here, when $c \ge \tfrac{2}{3}$, the rounding is
	deterministic: if $a < \tfrac{1}{3} \le b$, $ALG = 1$, $LP = b \ge
	\tfrac{1}{3}$; if $\tfrac{1}{3} \le a,b$, the cost to the algorithm is
	zero. Crucially, it is not possible to have both $a,b < \tfrac{1}{3}$,
	because $a+b \ge c \ge \tfrac{2}{3}$.

	Now we deal with the cases where $c < \frac{2}{3}$. The costs are
	\begin{align*}
		ALG &= [\fnk(c)(1-\fpk(a)) + \fpk(a)(1-\fnk(c))]_u
			  +[\fnk(c)(1-\fpk(b)) + \fpk(b)(1-\fnk(c))]_w\\
		LP &=[b(1-\fnk(c)\fpk(a))]_u
			  +[a(1-\fnk(c)\fpk(b))]_w
	\end{align*}
	When $a,b < \tfrac{1}{3}$, $ALG = 3c$, and $LP = a+b \ge c$, where we applied the triangle inequality.

	When $a < \tfrac{1}{3} \le b$, we have $ALG = 1$, $LP = b + a(1-c) \ge \tfrac{1}{3}$, by our assumption on $b$.

	For $\tfrac{1}{3}\le a,b$, $ALG = 2(1-c)$ and $LP = (b+a)(1-c) \ge \tfrac{2}{3}(1-c)$. Therefore, $3\cdot LP \ge ALG$ holds in all cases.

\medskip

\noindent\textbf{\mmn-Triangles:}
		Here, when $c \ge \tfrac{2}{3}$, the algorithm never incurs cost on
		unsafe edges. Therefore, we must simply calculate the costs when
		$c < \tfrac{2}{3}$:
	\begin{align*}
		ALG &= [(1-\fnk(c))(1-\fmk(a))]_u
		+[(1-\fnk(c))(1-\fmk(b))]_w\\
		&= (1-\tfrac{3}{2}c)(2-a-b),\\
		LP &=[(1-b)(1-\fnk(c)\fmk(a))]_u
		+[(1-a)(1-\fnk(c)\fmk(b))]_w\\
		&= (1-b)(1-\tfrac{3}{2}ac) + (1-a)(1-\tfrac{3}{2}bc)
	\end{align*}
		We verify that $3\cdot LP \ge ALG$:
		\begin{align*}
			LP - ALG= 3c(1-b)(1-a)\ge 0.
		\end{align*}
		Thus the approximation holds.

	This concludes the case analysis for $k$-partite complete graphs.
\end{proof}

\section*{Acknowledgements}
Part of this work was done when the
first, third, and fourth authors were visiting Microsoft
Research. Shuchi Chawla is supported in part by NSF Award CCF-1320854. Tselil Schramm is supported by an NSF
Graduate Research Fellowship, Grant No. DGE 1106400. Grigory
Yaroslavtsev is supported by an ICERM Institute Postdoctoral Fellowship at Brown and a Warren Center Postdoctoral Fellowship at Penn.

\bibliographystyle{abbrv}
\bibliography{corclustering}
\appendix
\section{Analysis of the Complete Case}
\label{app:complete}
In this section, we prove Theorem~\ref{thm:complete}. We follow the approach
outlined in Section~\ref{sec:performing_analysis}: We consider triangles $uvw$
of different type and for every type show that $\calC\equiv\alpha LP(uvw) -
ALG(uvw)\geq 0$ for all possible edge lengths $(x,y,z)$.  In this section, we
shall drop the argument $(uvw)$ of $LP$ and $ALG$; we let $\alpha = 2.06$. In
each of the cases below, the approximation ratio $\alpha = 2.06$ will not give
a completely tight analysis, giving us the $2.06 - \varepsilon$
approximation.\footnote{Some \ppp\ triangles and \ppm\ triangles with specific
edge lengths show that $\varepsilon < 0.01$.}

\subsection{\texorpdfstring{\mmm\ Triangles}{(-,-,-) Triangles}}
For triangles \mmm\ and \pmm, we obtain better approximations than $2.06$
\begin{lemma}
Fix $\fm(x) = x$. Then $ LP \ge ALG$ for all \mmm\ triangles.
\end{lemma}
\begin{proof}
Letting $\fm(x) = x$, the respective costs of the $ALG$ and $LP$ are
\begin{align*}
ALG
&= [(1-\fm(y))(1-\fm(z))]_u
+ [(1-\fm(x))(1-\fm(z))]_v
+ [(1-\fm(y))(1-\fm(x))]_w, \\
&= 3 - 2x - 2y - 2z + xy + xz + yz,\\
LP
&=
[(1-x)(1-\fm(y)\fm(z))]_u
[(1-y)(1-\fm(x)\fm(z))]_v
[(1-z)(1-\fm(x)\fm(y))]_w,\\
&= 3 - x - y - z - xy - xz - yz + 3xyz.
\end{align*}
Taking the difference,
\begin{align*}
LP - ALG
&= x + y + z - 2xy - 2yz - 2xz + 3xyz\\
&= x(1-y)(1-z) + y(1-x)(1-z) + z(1-x)(1-y)\ge 0.
\end{align*}
Hence, $LP \ge ALG$ for $x,y,z \in [0,1]$ as desired.
\end{proof}

\subsection{\texorpdfstring{\pmm\ Triangles}{(+,-,-) Triangles}}

For \pmm\ triangles, a simple condition on $\fp$ suffices to prove a 2-approximation.

\begin{lemma}
Fix $\fm(x) = x$. If $\fp(x) \le 2x$ for all $x \in [0,1]$, then $2\cdot LP \ge ALG$ for all $\pmm$-triangles.
\end{lemma}

\begin{proof}
In a $\pmm$-triangle with LP lengths $(x,y,z)$, we have costs
\begin{align*}
ALG &= [\fm(y)(1-\fm(z)) + \fm(z)(1-\fm(y))]_u
     + [(1-\fp(x))(1-\fm(z))]_v\\
     & \qquad + [(1-\fp(x))(1-\fm(y))]_w,\\
LP &= [x(1-\fm(z)\fm(y))]_u
     +[(1-y)(1-\fp(x)\fm(z))]_v
     +[(1-z)(1-\fp(x)\fm(y))]_w.
\end{align*}
Fixing $\fm(y) = y$, we have
\begin{align*}
2\cdot LP - ALG &= 2(2 + x -y - z -xyz - \fp(x)z - \fp(x) y + 2y\fp(x)z
\\ & \qquad- (2 - 2yz  -  2\fp(x) + \fp(x)z + \fp(x)y)\\
&= 2(1-y)(1-z) + 2x(1 -yz) + \fp(x)(2 - 3z - 3y + 4yz) \\
&\ge 2(1-y)(1-z) + 3\fp(x)(1 - z - y + yz) \\
&= 2(1-y)(1-z) + 3\fp(x)(1 - z)(1 - y) \\
&\ge 0,
\end{align*}
where to obtain the inequality in the third line we have applied the assumption
$2x \ge \fp(x)$. The statement follows.
\end{proof}

\begin{corollary}
For our choice of $\fp$, $2\cdot LP \ge ALG$ for all \pmm\ triangles.
\end{corollary}
\begin{proof}
We note that the function
\[
\fp(x) = \begin{cases} 0 & x < a \\
\tfrac{(x-a)^2}{(b-a)^2}  & a \le x \le b \\
1 & b < x\end{cases},
\]
is bounded above by $2x$ on the entire interval $[0,1]$; for $x < a$ and $x > b$ this is clear immediately, and for $x \in [a,b]$ we note that $\fp$ is convex and obtains value $1$ only at $b > \tfrac{1}{2}$.
\end{proof}

\subsection{\texorpdfstring{\ppm\ Triangles}{(+,+,-) Triangles}}
The expressions for \ppm\ and \ppp\ triangles are more complicated than
the expressions we dealt with above. We begin with the \ppm\ triangle.
Recall, that we defined $\calC = \alpha\cdot LP_{\ppm} - ALG_{\ppm}$.
For $\ppm$-triangles, using $\fm(z)=z$, we have
	\begin{align}
		\Cppm(x,y,z) &=
		\alpha\left(1 + x + y - z
 - yz\fp(x)
 - xz\fp(y)
 - (1-z)\fp(x)\fp(y)\right)\label{eq:def-Cppm}\\
 &\qquad	-\left(1 + 2z - 2z\fp(x) -2z\fp(y) + \fp(y)\fp(x)\right)\notag
 \end{align}

By Lemma~\ref{lem:tight-triangle}, it is sufficient to verify that $\Cppm(x,y,z)\geq 0$
only for all triangles with $x+y=z$ or $x+z=y$,
and the four triangles  $(b,b,1)$, $(a,1,1)$, $(b,1,1)$ and
$(1,1,1)$. For the latter four triangles we compute $\Cppm$ directly,
and check that $\Cppm\geq 0$.
From now on, we only consider $(x,y, x+y)$ and $(x,x+z,z)$ triangles.

We first focus on the case where the positive edge lengths both lie in
the range $[a,b]$. The following three lemmas argue that the worst
case for such triangles (i.e. a local minimum for $\Cppm$) is achieved
when the two positive edge lengths are equal.

\begin{lemma}
\label{lem:c_zero}
Let $2x \ge y$, so that the \ppm\ triangle $(x+c, x-c, y)$ obeys the triangle
inequality constraints. Then for $x \pm c \in [a,b]$,
\[
\tfrac{\partial}{\partial c} \Cppm(x+c, x-c, y) = 0 \ \  \text{ at }\ \ c = 0.
\]
For $x-c \in [0,a]$ and $x+c \in [a,b]$,
\[
\tfrac{\partial}{\partial c} \Cppm(x+c, x-c, y) \ge 0.
\]
\end{lemma}
\begin{proof}
	Consider \ppm triangle with LP lengths $(x+c, x - c, y)$. From (\ref{eq:def-Cppm}),
\begin{align*}
\alpha \cdot LP - ALG
&= \alpha(1 + 2x - y) - 1\\
& \qquad + y(2 + \alpha c - \alpha x)\fp(x+c) \ + y(2 - \alpha c - \alpha x)\fp(x-c)\\
&\qquad - (1+\alpha-\alpha y)\fp(x+c)\fp(x-c),
\end{align*}
We take a derivative with respect to $c$,
\begin{align*}
\tfrac{\partial}{\partial c}(\alpha \cdot LP - ALG)
&= \alpha y\fp(x+c) - \alpha y\fp(x-c)\\
&\qquad + y(2 + \alpha c - \alpha x)\tfrac{\partial}{\partial c}\fp(x+c) \ +
y(2 - \alpha c - \alpha x)\tfrac{\partial}{\partial c}\fp(x-c)\\
&\qquad- (1+\alpha - \alpha y)\fp(x+c)\tfrac{\partial}{\partial c}\fp(x-c)\ - (1+\alpha - \alpha y)\fp(x-c)\tfrac{\partial}{\partial c}\fp(x+c).
\end{align*}
Since $\fp$ is a quadratic in the range $[a,b]$, when $c = 0$ we have
$\tfrac{\partial}{\partial c} \fp(x + c) = - \tfrac{\partial}{\partial c}
\fp(x-c)$ when $c = 0$, and so at $c = 0$ this derivative is 0.
When $x - c \in [0,a]$ and $x + c \in [a,b]$, this derivative is
\begin{align*}
\tfrac{\partial}{\partial c}(\alpha \cdot LP - ALG)
&= \alpha y\fp(x+c) + y(2 + \alpha c - \alpha x)\tfrac{\partial}{\partial c}\fp(x+c),
\end{align*}
because $\fp(x-c) = 0$ and $\tfrac{\partial}{\partial c} \fp(x-c) = 0$. Since $\fp\ge 0$ and
$\alpha x < 2$, this quantity is positive, as desired.
\end{proof}

\begin{lemma}
\label{lem:c_real_roots}
Let $2x \ge y$, so that the \ppm\ triangle $(x+c, x-c, y)$ obeys the
triangle inequality constraints. Then for $x+c,x-c \in [a,b]$, the
degree-3 polynomial $\tfrac{\partial}{\partial c} \Cppm(x+c,x-c,y)$ as
a function of $c$ has three real roots: one positive, one negative,
and one zero.
\end{lemma}

\begin{proof}
	We compute the derivative $\tfrac{\partial}{\partial c} \Cppm(x+c,x-c,y)$
	and solve for its roots. One of the roots is zero, as shown in
	\prettyref{lem:c_zero}. The other two roots we compute to be
	\begin{align*}
		\pm\sqrt{
		\frac{
			(a-x)^2 +2y(1-\alpha \cdot a)(a-b)^2 + \alpha \left(
				(1-y)(x-a)^2
				+xy(b-a)^2\right)
	}{\alpha(1-y) + 1}},
	\end{align*}
	hence for $y \in [0,1]$, all roots of the partial derivative with respect to $c$ are real.
\end{proof}

\begin{lemma}
\label{lem:c_positive}
Let $2x \ge y$, so that the \ppm\ triangle $(x+c, x-c, y)$ obeys the triangle
inequality constraints. Then for $x+c,x-c \in [a,b]$, $\Cppm(x+c,x-c,y)$ is a
degree-4 polynomial in $c$ with a negative leading coefficient.
\end{lemma}

\begin{proof}
	We begin by examining \prettyref{eq:def-Cppm}, and
	noting that as $\fp(x+c)$ and $\fp(x-c)$ are quadratic in $c$,
	$\Cppm(x+c,x-c,y)$ is a degree-4 polynomial in $c$. The degree-4 terms come
	only from the product $\fp(x+c)\fp(x-c)$, so we compute the coefficients
	of this product in $\calC$ using \prettyref{eq:def-Cppm} to get the coefficients of
	$c^4$:
	\begin{align*}
		\frac{-(1 + \alpha(1-y))}{(b-a)^4} c^4,
	\end{align*}
	which is always negative for $y \in [0,1]$.
\end{proof}

\begin{corollary}[Lemmas \ref{lem:c_zero}, \ref{lem:c_positive}, \ref{lem:c_real_roots}]
	\label{cor:equal_sides}
	For $x,y \in [a,b]$, the function $\Cppm(x,y,z)$ has its only local minimum
	when $x = y$; that is, the most expensive \ppm triangles occur when the
	plus edges have equal length, or when $x = a$ or $y = b$.
\end{corollary}

\begin{lemma}
For our choices of $\fp, \fm$, $\Cppm(x,y,z) \ge 0$ for all $x,y,z \in [0,1]$.
\end{lemma}
\begin{proof}

  We begin with an analysis of $(x,y,x+y)$ triangles and consider
  several cases for the values of $x$ and $y$. Note that $x,y$ cannot
  both be larger than $b$, because that would make $z$ larger than
  $1$.

\smallskip
\noindent I. For $x,y \le a$, we have
	\begin{align*}
		\Cppm(x,y,z)
		&= \left(\alpha(x+y) - 2z\right) + \left(\alpha(1-z) - 1\right)\\
		&\ge (\alpha-2)z + \left(\alpha(1 -2a) - 1 \right)\ge 0,
	\end{align*}
	where in the second line we have twice applied the triangle inequality $2a \ge x+y \ge z$.

\medskip
\noindent II. When $x \le a, b \le y$, we have
	$$
		\Cppm(x,y,z) = \alpha(1+ x+y - z - xz) - 1
		\ge \alpha - \alpha xz - 1\ge 0.
	$$
	Here we applied the inequalities $x+y\ge z$ and $\alpha x\le \alpha a < 1$.

\medskip
\noindent III. When $a < x,y < b$, by \prettyref{lem:tight-triangle} and \prettyref{cor:equal_sides}, we may
	restrict our attention to \ppm-triangles with edge lengths $(x,x,2x)$; the other triangles
	from \prettyref{cor:equal_sides} are handled in other cases.
	For these triangles,
	\begin{align*}
		\alpha\cdot LP - ALG
		&= (\alpha - 1)  - (\alpha + 1)\fp(x)^2 + 4x(2\fp(x)-1) - 2\alpha x\fp(x)(2x-\fp(x)).
	\end{align*}
	Since $\fp$ is a degree-2 polynomial in $x$, we can solve for the roots of
	$\mathcal{C}$. The first root of $\mathcal{C}$ in the interval
	$[a,1]$ is  $0.500366$. Note that $x$ is at most $1/2$, so
        this root falls outside the interval of interest. Furthermore, at $x =
	a$, the polynomial is positive. Hence, we conclude that
	$\Cppm(x,y,z) \ge 0$ for all $x,y \in [a,b]$, $z = x+y$,
	as desired.

\medskip
\noindent IV. In the case $x < a < y < b$, \prettyref{lem:c_zero} allows us to conclude
	that $\Cppm(x,y,z)$ is minimized when $x = 0$, $y = z$.
	Calculating the roots of $\Cppm(0,y,y)$, the only real root is
	negative, and at the point $y = a$ the polynomial is positive, so
	$\Cppm(x,y,z) \ge 0$ for all $x \in [0,a]$, $y \in [a,b]$.

\medskip
\noindent V.
	In the final case of $a < x < b < y$, the function
	$\Cppm$ is
	strictly larger than in the case $a < x,y < b$, since increasing the length
	of positive edges past $b$ increases only the expected cost of the LP (see
	(\ref{eq:def-Cppm})).

\medskip
\noindent Next we analyze \ppm\ triangles with edge lengths
$(x,y,z)$ where $z=y-x$. First we note that if $y>b$, then decreasing this
length decreases $\Cppm$, and continues to satisfy the triangle
inequality. Next, if $x<a$ and $y\le b$, we can write:
\[\Cppm(x,y,z) = x^2\fp(y) + x(\alpha(2-y\fp(y)) + 2(1-\fp(y))) + (\alpha-1-2y+2y\fp(y)) \]
The first two terms in this expression are clearly non-negative. The
last term is non-negative because $\alpha-1\ge 2b\ge 2y$.

Finally, if $x,y\in [a,b]$, then we can decrease $y$ by an amount $c$
and increase $x$ by $c$, keeping $z$ the same, and apply
\prettyref{lem:c_zero} to argue that $\Cppm$ decreases. Therefore, it
remains to consider edge lengths $(x,x,0)$. In this case,
\begin{align*}
	\Cppm &= \alpha(1+2x-\fp(x)^2) - (1+\fp(x)^2)\\
				   &= (\alpha - 2) + (1 - \fp(x)^2) + \alpha(2x-\fp(x)^2),\\
\end{align*}
and this expression is positive, since $2x-\fp(x)^2\ge 2x-\fp(x)\ge 0$,
$1-\fp(x)^2 \ge 0$, and $\alpha > 2$.
\end{proof}

\subsection{\texorpdfstring{\ppp\ Triangles}{(+,+,+) Triangles}}
As in the previous cases, we write $\calC$ as a function of edge lengths $(x,y,z)$:
 \begin{align*}
\Cppp(x,y,z) &= \alpha\left(
x + y + z - y\fp(z)\fp(x)
 - x\fp(z)\fp(y)
 -z\fp(x)\fp(y)
	\right) \label{eq:def-Cppp}\\
	&\qquad - 2\left(\fp(x) + \fp(y) + \fp(z) - \fp(y)\fp(x) - \fp(y)\fp(z) - \fp(z)\fp(x)\right)\notag
 \end{align*}

By Lemma~\ref{lem:tight-triangle}, it suffices to check that $\Cppp(x,y,z)\geq
0$ only for triangles with $z=x+y$. As before, we can verify $\Cppp\geq 0$ for
several corner cases, given in  Lemma~\ref{lem:tight-triangle}, by direct
computation.

We focus first on the case where all $x,y,z$ lie in the range $[a,b]$.

\begin{lemma}
	\label{lem:neg_c}
	Let $y,x \pm c \in [a,b]$. The polynomial $\Cppp(x+c,x-c,y)$ is a degree-4
	polynomial in $c$ with positive leading coefficient.
\end{lemma}
\begin{proof}
	We begin by examining the above expression for $\Cppp$, and
	noting that as $\fp(x+c)$ and $\fp(x-c)$ are quadratic in $c$,
	$\Cppm(x+c,x-c,y)$ is a degree-4 polynomial in $c$. The degree-4 terms come
	only from the product $\fp(x+c)\fp(x-c)$, so we compute the coefficients
	of this product in $\calC$ using \prettyref{eq:elp} and
	\prettyref{eq:ecost} to get the coefficients of
	$c^4$:
	\begin{align*}
		\frac{(2 - \alpha y)}{(b-a)^4} c^4,
	\end{align*}
	which is always positive for $y \in [a,b]$.
\end{proof}

\begin{lemma}
\label{lem:ppp_real_roots}
Let $2x \ge y$, so that the \ppp\ triangle $(x+c, x-c, y)$ obeys the triangle
inequality constraints. Then for $x\pm c, 2x \in [a,b]$,
$\tfrac{\partial}{\partial c} \Cppm(x+c,x-c,y)$ has all real roots in $c$.
\end{lemma}

\begin{proof}
	We compute the derivative $\tfrac{\partial}{\partial c} \Cppm(x+c,x-c,y)$
	and solve for its roots. One of the roots is zero, by a similar calculation
	to the \ppm\ case shown in \prettyref{lem:c_zero}. The other two roots are
	\begin{align*}
		\pm\sqrt{
		\frac{
			(2 - \fp(y) ( 2  - \alpha 2a +\alpha x) )(a-b)^2 + (2-\alpha y)(a-x)^2
	}{2-\alpha y}}.
	\end{align*}
	To argue that these roots are real for $x\pm c,y \in [a,b]$, it now suffices
	to check that $(2 - \fp(y) ( 2  - \alpha 2a +\alpha x) ) \ge 0$;
	\begin{align*}
		(2 - \fp(y) ( 2  - \alpha 2a +\alpha x) )
		&\ge 2 - \fp(2x)(2 - \alpha 2a + \alpha x)\\
		&\ge 2 - \fp(b)(2 - \alpha 2a + \alpha \tfrac{b}{2})\\
		&\approx 0.25
	\end{align*}
	where in the first line we have applied the triangle inequality on
	$x+c,x-c,y$, and in the second line we apply the restriction of $y \in
	[a,b]$. Hence, we conclude that all three roots of
	$\tfrac{\partial}{\partial c}\Cppp(x+c,x-c,y)$ are real for $y,x\pm c \in
	[a,b]$.
\end{proof}

\begin{corollary}[Lemmas \ref{lem:tight-triangle}, \ref{lem:neg_c}, and \ref{lem:ppp_real_roots}] \label{cor:ppp_worst}
	The local minima of the polynomial $\Cppp(x_{uv},x_{vw},x_{uw})$ over
	$x_{uv},x_{vw},x_{uw} \in [a,b]$ when
	\begin{align*}
		x_{uv} &= x + r(x),\\
		x_{vw} &= x - r(x),\\
		x_{uw} &= 2x,
	\end{align*}
	where $r(x)$ is a nonzero root of $\tfrac{\partial}{\partial c} \Cppp(x+c,x-c,2x)$,
	\[
		r(x) =
		\sqrt{
		\frac{
			2(1-\fp(2x))(a-b)^2 + (2- 2 \alpha x)(a-x)^2
			+\alpha \fp(2x)( 2 a - x)(a-b)^2
	}{2-2 \alpha x }}.
	\]
\end{corollary}

We are now in a position to prove that triangles with all side lengths in the
range $[a,b]$ give an $\alpha = 2.06$ approximation.

\begin{lemma}
	\label{lem:ppp_all_in}
	Consider a \ppp\ triangle with edge lengths $(x_{uv},x_{vw},x_{uw})$ such
	that $x_{uv},x_{vw},x_{uw} \in [a,b]$ satisfy the triangle inequality. Then
	$\Cppp(x_{uv},x_{vw},x_{uw})\ge 0$ in this range.
\end{lemma}
\begin{proof}
	By \prettyref{cor:ppp_worst}, the minimizing triangle is a triangle with
	lengths $\left(x+r(x),x-r(x),2x\right)$. Furthermore, we need only consider
	$x$ such that $x \pm r(x),2x \in [a,b]$.

	The first bound we get from these restrictions is that $0.095 = \tfrac{a}{2}
	\le x \le \tfrac{b}{2} \approx 0.255 $. We then solve for $x - r(x) \ge a$,
	and we have the restriction $x \ge .265$. Therefore, there are no triangles
	of the minimizing form that adhere to the restrictions $x\pm r(x),2x \in
	[a,b]$.

	Thus, it suffices to check triangles of the form $x\pm c, 2x \in
	[a,b]$ where either $x-c = a$ or $x+c = b$, since
	as a consequence of \prettyref{cor:ppp_worst} the minimizing triangles are
	on the boundary of the interval for $x \pm c$.

	The first boundary case is a triangle with side lenghts $(a,2x-a,2x)$. In
	this case,
	\begin{align*}
		\Cppp(a,2x-a,2x)
		&\approx -2.73172+27.4989 x+133.045 x^2-1407.86 x^3+2469.93 x^4
	\end{align*}
	Solving for the roots in $x$, we get $x \approx -
	0.1325,0.0917,0.2570,0.3537$. Since the leading coefficient is positive,
	$\Cppp$ is positive between $x = 0.0917$ and $x = 0.2570$; from our
	restriction $2x \le b$ we have $x < 0.257$, and from $2x-a \ge a$ we have
	$x \ge a > 0.0917$. Hence, $\Cppp(a,2x-a,2x) > 0$.

	The second boundary case requires that $c = b-x$, so that the edge of
	length $x-c = 2x - b$, and there cannot be a triangle with side lengths
	$(2x-b,b,2x)$ such that both $2x, 2x - b \in [a,b]$.

	Hence, in all of the boundary cases $\Cppp(x_{uv},x_{vw},x_{uw}) \ge 0$, as desired.
\end{proof}

The case in which $x < a$, $y,y+x \in [a,b]$ is also a bit more involved; we analyze it
separately.
\begin{lemma}
	\label{lem:ppp-less-in-in}
	Consider a \ppp\ triangle with edge lengths $(x_{uv},x_{vw},x_{uw})$ such
	that $x_{uv} < a$, $x_{vw},x_{uw} \in [a,b]$  satisfy the triangle inequality. Then
	$\Cppp(x_{uv},x_{vw},x_{uw})\ge 0$ in this range.
\end{lemma}
\begin{proof}
	We first appeal to \prettyref{lem:tight-triangle}, and consider only
	triangles such that $x_{uw} = x_{uv} + x_{vw}$.
	Again, we set $x_{uv} = x - c$, $x_{vw} = x+c$, and $x_{uw} = 2x$, so that
	we may take derivatives with respect to the difference in lengths of the two shorter edges.

	Consider $\Cppp$,
	\begin{align*}
	\Cppp(x-c,x+c,2x)
	&=
	4 \alpha x - 2\fp(x + c) + \fp(2x)((2 + \alpha c - \alpha x)\fp(x + c)-2)
\end{align*}
Expressed as a polynomial in $c$, $\Cppp$ is a degree-3 polynomial with
positive leading coefficient $\alpha\fp(y)/(b-a)^2$. Thus, the second root of
the derivative with respect to $c$ is a local minimum.  Taking a derivative
with respect to $c$, then solving for the second root, we get
\begin{align*}
	r(x) & = \frac{4b(b-2a) + \alpha a^3 + 16 a x - 3\alpha a^2 x - 16{x}^2 + 4\alpha {x}^3 }{3\alpha(2x-a)^2}.
\end{align*}
We now consider the values of $x$ for which $x - r(x) \ge 0$,
\begin{align*}
	x &\ge r(x),
	3x\alpha(2x - a)^2 -\Big( 4b(b-2a) + \alpha a^3 + 16 a x - 3\alpha a^2 x - 16{x}^2 + 4\alpha {x}^3\Big)  &\ge 0.
\end{align*}
Solving for roots of this cubic in $x$, we have the first non-negative root as
$x \approx 0.2438$. Since the coefficient of ${x}^3$ is $8\alpha > 0$, this
implies that for all $x < 0.2438$, the minimizing value of $c$ is one of the
boundary cases, $x-c = 0$ or $x-c = a$; we will verify these boundary cases
later.

Otherwise, we compute the value of $\Cppp(x-r(x),x+r(x),2x)$. After some simplification,
we get
\begin{align*}
	\Cppp(x-r(x),x+r(x),2x)
	&=-\frac{59980.5}{(0.19-2x)^6} (x-0.254936) (x -0.237264) (x+0.132462)\\
 & \qquad \times \left(0.739503-1.53615 x+{x}^2\right) \left(0.079347-0.452156 x+{x}^2\right)\\
	&\qquad \times \left(0.009025-0.19 x+{x}^2\right) \left(0.000908218+0.0367349 x+{x}^2\right)
\end{align*}
The real non-negative roots of this expression are $x = 0.237264, 0.254936$,
the latter of which is larger than $\tfrac{b}{2}$. For $x \in [0.237264,
0.254936]$, the expression is positive, and this interval contains all $x$
such that $x - r(x) \ge 0$ and $2x \le b$, we conclude that in these cases
$\Cppp$ is positive.

In the case where $x - r(x) < 0$, it remains to check at the boundaries $x =
0,a$ that $\Cppp(0, y, y) \ge 0$ and $\Cppp(a,y,a+y) \ge 0$. The latter case we
have already verified in \prettyref{lem:ppp_all_in}. To check the former case,
we calculate
$\Cppp(0,y,y)$:
\begin{align*}
	\Cppp(0,y,y)
	&= 2(\alpha y - 2\fp(y) + \fp(y)^2)
\end{align*}
Solving for roots in $y$, we have $y = -0.29, y = 0.09$ as roots. The leading
coefficient is $\frac{2}{(a-b)^4}$, so for $y \in [a,b]$, $\Cppp(0,y,y)$ is
positive.
\end{proof}

We now complete the argument by considering all the other cases.

\begin{lemma}
	The quantity $\Cppp(x,y,z) \ge 0$ for all $x,y,z \in [0,1]$.
\end{lemma}
\begin{proof}

	We now give an analytical proof that $\Cppp\geq 0$.	By
	\prettyref{lem:tight-triangle}, we may consider only the cases in which
	$x+y = z$. We first deal with several simple cases.

\noindent I. For $z>b$, we note that the cutting
 probability $\fp(z) = \fp(b) = 1$, so $\Cppp$ can be written as a
 term independent of $z$ plus the term $\alpha (1-p_x p_y)z$, which is
 positive. So, as we gradually decrease $z$, $\Cppp$ decreases and in the
 worst case we have $z=b$. This is covered by the remaining cases below.

\smallskip
\noindent II. For $x,y< a$, we get
	$$
		\Cppp(x,y,z)
		= \alpha(x+y+z) - 2\fp(z)
		\ge 2\alpha z - 2\fp(z),
	$$
	where we have applied the triangle inequality, and since $\fp(z) < 2z$,
	this quantity is positive.

\medskip
 \noindent III. 	In the case $a < x \le y,z \le b$,
we simply appeal to \prettyref{lem:ppp_all_in}.

\medskip
\noindent IV. In the case $x < a \le y,z \le b$, we appeal to \prettyref{lem:ppp-less-in-in}.

This completes our case analysis.
\end{proof}

\end{document}